\documentclass[12pt,onecolumn,draftcls]{IEEEtran}
\usepackage{graphicx}
\usepackage{bm}
\usepackage[cmex10]{amsmath}
\usepackage{amssymb}
\usepackage{acronym}
\usepackage{setspace}
\usepackage{amsthm}
\usepackage{color}
\usepackage{footnote}
\usepackage{cite}

\newtheorem{theorem}{Theorem}
\newtheorem{lemma}{Lemma}

\graphicspath{{figure/}}

\ifCLASSINFOpdf

\else

\fi
\begin{document}

\title{Blind Signal Detection in Massive MIMO: Exploiting the Channel Sparsity}
\author{Jianwen~Zhang,
		Xiaojun~Yuan,~\IEEEmembership{Senior Member,~IEEE},
		and Ying Jun (Angela) Zhang,~\IEEEmembership{Senior Member,~IEEE}
\thanks{X. Yuan is with the National Key Laboratory of Science and Technology on Communications, University of Electronic Science and Technology of China, Chengdu, China. Email: xjyuan@uestc.edu.cn.}
\thanks{Y. Zhang is with the Department of Information Engineering, the Chinese University of Hong Kong, Shatin, New Territories, Hong Kong. Email: yjzhang@ie.cuhk.edu.hk.}
		}

\maketitle

\begin{abstract}
	 In practical massive MIMO systems, a substantial portion of system resources are consumed to acquire channel state information (CSI), leading to a drastically lower system capacity compared with the ideal case where perfect CSI is available. In this paper, we show that the overhead for CSI acquisition can be largely compensated by the potential gain due to the sparsity of the massive MIMO channel in a certain transformed domain. To this end, we propose a novel blind detection scheme that simultaneously estimates the channel and data by factorizing the received signal matrix. We show that by exploiting the channel sparsity, our proposed scheme can achieve a DoF very close to the ideal case, provided that the channel is sufficiently sparse. Specifically, the achievable degree of freedom (DoF) has a fractional gap of only $1/T$ from the ideal DoF, where $T$ is the channel coherence time. This is a remarkable advance for understanding the performance limit of the massive MIMO system. We further show that the performance advantage of our proposed scheme in the asymptotic SNR regime carries over to the practical SNR regime. Numerical results demonstrate that our proposed scheme significantly outperforms its counterpart schemes in the practical SNR regime under various system configurations.
\end{abstract}

\IEEEpeerreviewmaketitle

\section{Introduction}
	\IEEEPARstart{M}{assive} multiple-input multiple-output (MIMO) systems have been extensively studied in the past decade for its advantages on boosting the system throughput, improving the link reliability, and enhancing the energy efficiency \cite{Larsson14,Marzetta10,Ngo13,Hoydis13,Geraci13,Fang16}. Consider a massive MIMO system consisting of $K$ single-antenna transmit terminals and a receive terminal with $N$ antennas, under the massive MIMO assumption of $N \gg K \gg 1 $. This setting arises in machine type communication scenarios as specified by 5G system requirements \cite{ITU,Wang}, where many low-complexity terminals (devices) with a single antenna need to communicate with a powerful base station with an array of multiple antennas. A fundamental problem for massive MIMO is to determine the system capacity. It is well known that ideally, when the MIMO channel matrix is perfectly known to the receiver, the capacity of the system scales as $K\log(\mathsf{SNR})$ at high SNR, i.e., the degrees of freedom (DoF) of the system is $K$. However, in practical systems, the acquisition of channel state information (CSI) consumes a substantial amount of system resource. The system capacity of the ideal case is therefore difficult to achieve from a practical viewpoint. In fact, as the MIMO size becomes large, the system overhead spent on CSI acquisition increases and eventually becomes the bottleneck to increase the system capacity.
	
	There are two canonical research directions for channel acquisition. The first direction is referred to as the training-based approach, in which each transmission frame is divided into two phases, namely, the training phase and the data transmission phase \cite{Coldrey07,Yuan16}. In the training phase, the transmitters transmit pilot signals and the receiver estimates the channel coefficients based on the knowledge of the pilot signals. In the data transmission phase, the transmitters transmit data, and the receiver detects the data based on the estimated channel. In the training-based approach, a pilot length of no less than $K$ is required to probe the channel with a vanishing estimation error \cite{Yuan16,Hassibi03}. This leads to a DoF of $K(1-\frac{K}{T})$, where $T$ is the channel coherence time, and the DoF loss compared to the ideal case is due to the fact that no information is carried by the pilot signals. To avoid the training overhead, another line of research works on blind detection, in which the receiver estimates the channel and detects data without any prior knowledge of the signals from the transmitters \cite{Muquet02,Ngo12,Zheng02}. However, as the gain from no training overhead is largely compromised by the reduction of detection accuracy due to channel uncertainty, blind detection achieves the same DoF as the training-based approaches, i.e., the DoF for blind detection is still given by $K(1-\frac{K}{T})$ \cite{Zheng02}.
	
	The aforementioned approaches assume a rich-scattering multipath environment and so the channel coefficients can be modelled as random variables satisfying a certain continuous distribution. This assumption, however, is questionable in massive MIMO systems. More and more analyses and experimental evidences demonstrate that the physical channel of a massive MIMO system exhibits a sparse structure in the angular domain of the receive antenna array \cite{Sayeed07,Sayeed02,Samimi16,Vuokko07,Czink07}, i.e., the channel coefficient matrix has many zero or near-zero elements. The reason is two-fold. On one hand, a growing demand for bandwidth increases radio frequency and reduces wavelength, while an electromagnetic wave with a shorter wavelength is more likely to be blocked by obstructions. As a result, there will be fewer propagation paths in the channel for next-generation wireless communications. On the other hand, with the deployment of large-scale antenna arrays, the resolution bin in the angular domain becomes much finer than ever before. This enables the receiver to distinguish the angles of arrival for different paths with a much higher resolution.
	
	The channel sparsity can be exploited to enhance the performance of a massive MIMO system \cite{Yin13,Bajwa10,Lau14,Masood15,Muller14, Mezghani16}. For example, in training-based massive MIMO systems, compressed sensing was used to reduce the number of required pilot signals by exploiting the channel sparsity \cite{Bajwa10,Lau14,Masood15}. It has been shown that compressed-sensing based training schemes can achieve a DoF of $K(1-c\frac{K}{T})$, where $c$ is usually a coefficient between 0 and 1 depending on the channel sparsity level. Moreover, the channel sparsity has been utilized in blind channel estimation \cite{Muller14, Mezghani16}. The basic idea is to approximately calculate the receive covariance matrix using the received signal, and then to estimate the channel matrix by factorizing the approximate covariance matrix based on the sparsity of the channel matrix. Afterwards, the data is detected based on the estimated channel. The above blind channel estimation scheme has the benefit of avoiding the pilot overhead. However, to obtain a relatively accurate estimate of the receive covariance matrix, the coherence time $T$ of the channel is required to satisfy $T \gg N $, which is unfortunately difficult to realize in a massive MIMO system. Thus, the performance of the blind channel estimation scheme is quite poor for a massive MIMO system, especially in a block-fading environment with a relatively short coherence time $T$.
	
	In this paper, we investigate the impact of the channel sparsity on the fundamental performance limit of a massive MIMO system. Specifically, \emph{we propose a novel blind massive MIMO detection scheme that simultaneously estimates the channel and detect the signal from the received signal by exploiting the channel sparsity}. Unlike the blind channel estimation scheme in \cite{Mezghani16}, our proposed blind detection scheme does not rely on an accurate estimation of the receive covariance matrix, and can work well even when $T < N$. \emph{We show that, with the channel sparsity and under some regularity conditions, our scheme can achieve a DoF arbitrarily close to $K(1-\frac{1}{T})$ for a sufficiently large $N$ and the channel is sufficiently sparse}. This implies a huge throughput improvement of the massive MIMO system over the existing approaches \cite{Bajwa10,Lau14,Masood15,Muller14, Mezghani16} in the high SNR regime. In addition, the DoF of our scheme is very close to the ideal DoF of $K$, implying that the adverse effect of channel uncertainty can be largely compensated by the potential gain due to the channel sparsity.
	
	We further consider the algorithm design for the blind detection scheme to achieve the potential gain of the channel sparsity in the practical SNR regime. We point out that the blind signal detection problem under concern is related to dictionary learning \cite{Sun17} and sparse matrix factorization \cite{Koren09}. Specifically, the joint estimation of the channel and the data based on the received signal can be formulated as a sparse matrix factorization problem. This problem is non-convex and so is difficult to find an optimal solution. There exist a number of approximate solutions in the literature, such as the K-SVD algorithm \cite{KSVD06}, the SPAMS algorithm \cite{SPAMS10}, the ER-SpUD algorithm \cite{Spielman12}, and the bilinear generalized approximate message passing (BiG-AMP) algorithm \cite{PSchniter14}. Among these algorithms, BiG-AMP is known to have the best performance in general. However, we show that BiG-AMP does not work well for $T > K$, which is a typical setting in a massive MIMO system, when BiG-AMP is directly applied to our blind detection problem. To address this issue, \emph{we propose a projection-based BiG-AMP (P-BiG-AMP) algorithm}, in which the subspace occupied by the signal is estimated in the first place, and then BiG-AMP is applied to factorize the image of the received signal projected onto the estimated signal subspace. \emph{Numerical results demonstrate that our proposed blind detection scheme with P-BiG-AMP significantly outperforms the counterpart schemes in the practical SNR regime under various configurations of $N$, $K$, and $T$}.

\subsection{Organization}	
	The remainder of this paper is organized as follows. In Section II, we introduce a sparse channel model for the massive MIMO system. In Section III, we present upper and lower bounds of the capacity of the massive MIMO system. In Section IV, we analyze the DoF of the proposed scheme. In Section V, we develop a message-passing based detection algorithm to jointly detect the data of users and the channel matrix. Numerical results are presented in Section VI to verify the effectiveness of our proposed scheme. Finally, we conclude the paper in Section VII.

\subsection{Notation}
	Regular letters, lowercase bold letters, and capital bold letters represent scalars, vectors, and matrices, respectively. $\mathbb{C}$ denotes the complex field; the superscripts $(\cdot)^\text{H}$, $(\cdot)^*$, $(\cdot)^\text{T}$, and $(\cdot)^{-1}$ represent the conjugate transpose, conjugate, transpose, and the inverse of a matrix, respectively; $|\cdot|,\|\cdot\|_1, \|\cdot\|_2$, and $\|\cdot\|_F$ represent the absolute value, the $\ell_1$-norm, the $\ell_2$-norm, and the Frobenius norm, respectively; $\mathsf{E}[\cdot]$, $\det(\cdot)$, and $\log(\cdot)$ represents the expectation, the determinant, and the logarithm function; $\text{diag}\{\mathbf{a} \}$ represents the diagonal matrix with the diagonal specified by $\mathbf{a}$; $\lceil a \rceil$ represents the minimum integer larger than $a$. For an integer $N$, $\mathcal{I}_N$ denotes the set of integers from $1$ to $N$. The notation $\propto$ denotes equality up to a constant scaling factor. The notation $a \lesssim b$ means $\limsup_{b\rightarrow\infty} \frac{a}{b} \leq c$, where $c>0$ is a constant. Similarly, $a \gtrsim b$ means $\limsup_{a\rightarrow\infty} \frac{b}{a} \leq c$.

\section{System Model}
\subsection{Sparse Channel Modeling}
	We now present a sparse channel model for massive MIMO systems by following the approach in \cite{Bajwa10}. Consider a massive MIMO channel with $K$ single-antenna transmitters and a receiver deployed with uniform linear array (ULA) of $N$ antennas. By massive MIMO, we assume $N \gg K \gg 1 $. This setting arises in practical scenarios, e.g., when a base station that deploys an array of a few hundred antennas communicates with tens of users. Let $L_r$ be the normalized length of the ULA.\footnote{This means, the actual length of the ULA is $\lambda_c L_r$, where $\lambda_c$ is the wavelength.} Then, the normalized interval between any two adjacent receive antennas is $\Delta_r = L_r/N$. Denote by $L_k$ the number of physical paths between transmitter $k$ and the receiver, by $\alpha_{l,k}$ the path gain of the $l$th path of transmit $k$, and by $\theta_{l,k}$ the AoA of the $l$th path of transmitter $k$. Then, the physical channel of a nonselective MIMO channel from transmitter $k$ to the receiver can be modeled by
	\begin{equation}\label{equ:phychn}
		\tilde{\mathbf{h}}_k = \sum_{l = 1}^{L_k} \alpha_{l,k} \mathbf{a}_{r} (\theta_{l,k}),
	\end{equation}
	where 
	\begin{align}
		\mathbf{a}_r(\theta_{l,k}) = \frac{1}{\sqrt{N}} \begin{bmatrix}
															1\\
															\exp(-\mathsf{j}2\pi\Delta_r \cos \theta_{l,k} )\\
															\vdots\\
															\exp(-\mathsf{j}2\pi(N-1)\Delta_r \cos \theta_{l,k})
														\end{bmatrix}
	\end{align}
	represents the array steering vector for receiving a signal from transmitter $k$ in the direction given by $\theta_{k,l}$, and $\mathsf{j} = \sqrt{-1}$. By using the virtual representation method in \cite{Bajwa10}, we can rewrite \eqref{equ:phychn} as
	\begin{equation}
		 \tilde{\mathbf{h}}_k =  \sum_{n = 0}^{N-1} h_{n,k} \mathbf{a}_{r} \left(\arccos \frac{n}{L_r} \right)   = \mathbf{A}_r \mathbf{h}_k, \label{equ:H_tield}
	\end{equation} 
	where $\mathbf{A}_r = \left[\mathbf{a}_{r} \left(\arccos \frac{0}{L_r} \right),\ldots,\mathbf{a}_{r} \left(\arccos \frac{N-1}{L_r} \right) \right] \in \mathbb{C}^{N \times N}$ is a unitary matrix, and $\mathbf{h}_k$ is the $k$th column of $\mathbf{H}$ in \eqref{equ:System1}. From \eqref{equ:H_tield}, $\mathbf{h}_k = [h_{1,k}, \cdots, h_{N,k}]^\text{T}$ can be treated as an equivalent channel of user $k$ in the angular domain, where $h_{n,k}$ is the aggregated gain of the physical paths of user $k$ within the resolution bin centered around $\arccos \frac{n}{L_r}$ in the angular domain. Denote by $\tilde{\mathbf{H}} = [\tilde{\mathbf{h}}_1, \cdots, \tilde{\mathbf{h}}_K]$ the overall channel matrix. Then
	\begin{equation}
		\tilde{\mathbf{H}} = \mathbf{A}_r \mathbf{H},
	\end{equation}
	where $\mathbf{H} = [\mathbf{h}_1, \cdots, \mathbf{h}_K]$ is the projection of the channel in the angular domain. 
	
	We now describe the sparsity of $\mathbf{H}$ for massive MIMO systems. Recall that each $h_{n,k}$ is the aggregated channel gain of the physical paths of transmit terminal $k$ whose AoAs are within the $n$th resolution bin in the angular domain. It has been previously discussed in \cite{Sayeed07,Sayeed02,Vuokko07,Czink07} that a large portion of the elements of $\mathbf{H}$ are very close to zero, since the number of resolution bins of a massive antenna array usually far exceeds the number of physical paths of each user $k$, i.e., $N \gg L_k$. An example is illustrated in Fig. \ref{Fig:ChnSpasity}. These near-zero elements of $\mathbf{H}$ correspond to weak channel links, and can be essentially ignored in the transceiver design of a massive MIMO system. 
	
	Based on the above discussions, we henceforth assume that the massive MIMO channel is sparse, i.e., the sparsity level $\rho$ satisfies
	\begin{equation}\label{equ:sparsity_level}
		\rho = \frac{|\mathcal{S}| }{NK} < 1,
	\end{equation}
	where $\mathcal{S}$ is the support of the non-zero elements of $\mathbf{H}$, i.e.
	\begin{equation}
		h_{n,k} = 0,\ \ \text{for} \ (n, k)\notin \mathcal{S},
	\end{equation}
	and $|\mathcal{S}|$ represents the cardinality of the set $\mathcal{S}$.
	
	\begin{figure}[!t]
		\centering
		\includegraphics[width = 0.7\textwidth]{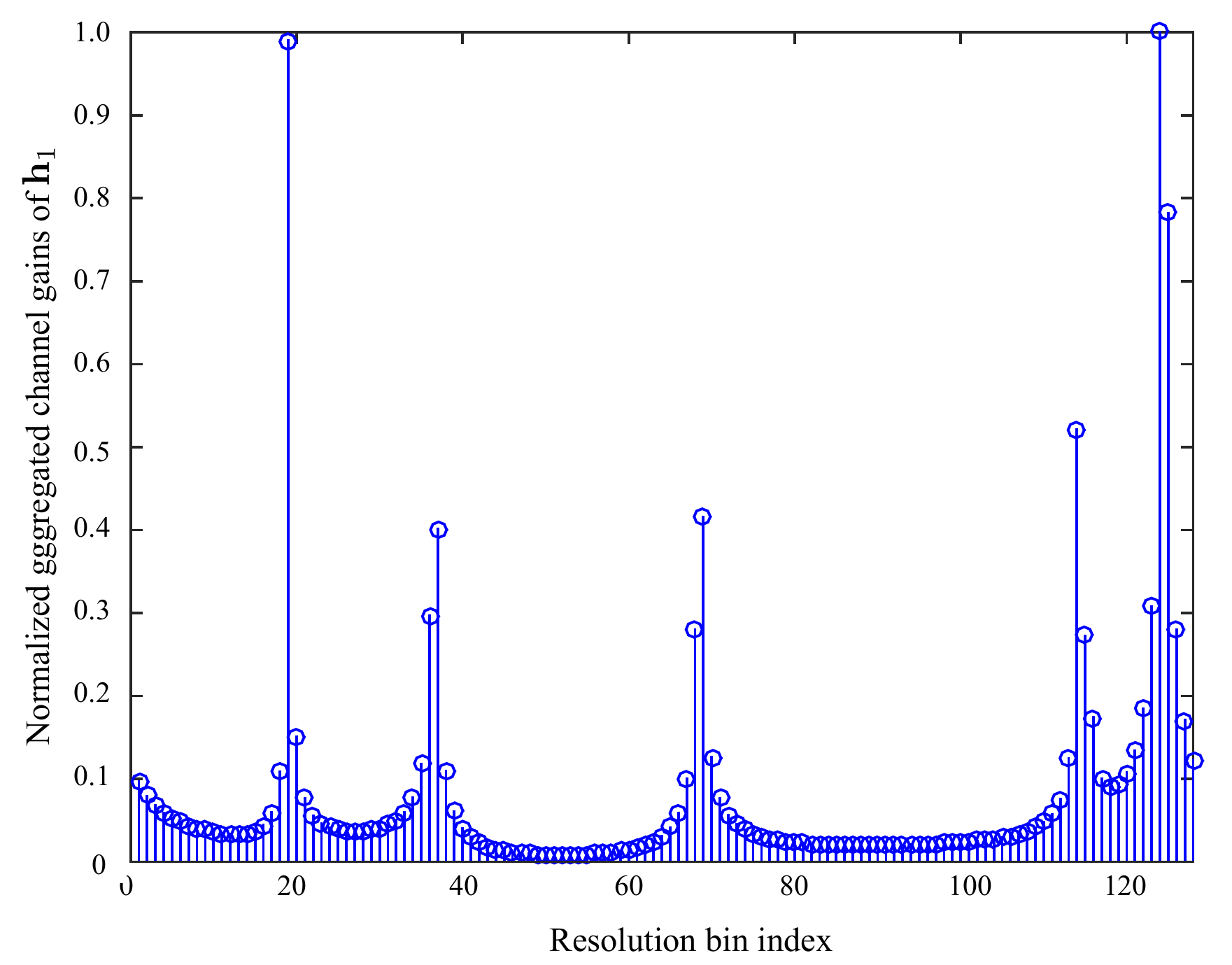}
	 	\caption{An example to illustrate the channel sparsity of transmitter 1 in the angular domain. The number of received antenna is $N = 128$. Half-wavelength separation is assumed between any two of the receive antennas, i.e., $\Delta_r = 0.5$. The number of physical paths is $L_1 = 6$. The path gains $\{\alpha_{l,k}\}$ are drawn from the circularly symmetric complex Gaussian distribution with zero mean and unit variance. The AoAs $\{\theta_{l,k}\}$ are randomly and uniformly distributed between $0$ and $\pi$. Particularly, in the figure, $[\alpha_{1,1}, \alpha_{2,1}, \alpha_{3,1}, \alpha_{4,1}, \alpha_{5,1}, \alpha_{6,1}] = [0.2963+\mathsf{j}0.5249,\ -0.1011-\mathsf{j}1.5287,\ -0.4555+\mathsf{j}1.0179,\ -0.5989+\mathsf{j}0.0405,\ 0.4550-\mathsf{j}0.4741,\ -0.0022+\mathsf{j}0.2496]$, and $[\theta_{1,1},\theta_{2,1},\theta_{3,1},\theta_{4,1},\theta_{5,1},\theta_{6,1}] = [0.3123\pi,\ 0.5227\pi,\ 0.4086\pi,\ 0.8929\pi,\ 0.5738\pi,\ 0.5679\pi]$. The maximum value of $|h_{1j}|$ is normalized to 1. }\label{Fig:ChnSpasity}
	\end{figure}
	
\subsection{Signal Model}
	The signal model of the massive MIMO system is presented as follows. The channel is assumed to be block-fading, i.e., the channel remains unchanged within the coherence time $T$. Then, for each transmission block of duration $T$, the received signal matrix is represented by
	\begin{equation} \label{equ:Y_tield}
		\tilde{\mathbf{Y}} = \tilde{\mathbf{H}} \mathbf{X} + \tilde{\mathbf{W}} 
		                   = \mathbf{A}_r \mathbf{H} \mathbf{X} +  \tilde{\mathbf{W}},
	\end{equation} 
	where $\tilde{\mathbf{Y}} \in \mathbb{C}^{N\times T}$ is the received signal over $T$ time slots, $\mathbf{X} \in \mathbb{C}^{K\times T}$ is the transmitted signal from all the $K$ transmit terminals, and $\tilde{\mathbf{W}} \in \mathbb{C}^{N\times T} $ is the additive white Gaussian noise with each element independently drawn from $\mathcal{CN}(0,\sigma^2)$. 
	
	By left-multiplying $\tilde{\mathbf{Y}}$ with $\mathbf{A}_r^\text{H}$, we obtain  the projection of the received signal in the angular domain as 
	\begin{equation} \label{equ:System1}
		\mathbf{Y} = \mathbf{HX} + \mathbf{W},
	\end{equation}
	where $\mathbf{Y} =\tilde{\mathbf{Y}}\mathbf{A}_r^\text{H} \in \mathbb{C}^{N\times T}$, and $\mathbf{W}=\tilde{\mathbf{W}}\mathbf{A}_r^\text{H} \in \mathbb{C}^{N\times T}$. Note that the elements of $\mathbf{W}$ are still independent and identically distributed Gaussian random variables with zero mean and variance $\sigma^2$. Let $\mathbf{x}_k \in \mathbb{C}^{T\times 1} $ be the transpose of the $k$th row of $\mathbf{X}$. We assume that the average transmission power of each transmit terminal is given by $\alpha_k P$, i.e.
	\begin{equation}\label{equ:power_constraint}
		\frac{1}{T}\mathsf{E}[\mathbf{x}_k^\text{H} \mathbf{x}_k] \leq \alpha_k P, \text{ for all } k \in \mathcal{I}_K \triangleq \{1,2,\cdots, K\},
	\end{equation}
	where $\alpha_k \geq 0$ for $k \in \mathcal{I}_K$ satisfy $\sum_{k=1}^K \alpha_k = 1$, and $P$ is the total power budget.
	
	The system capacity of a massive MIMO system is given by
	\begin{equation}\label{equ:capacity}
		C(\mathsf{SNR}) = \frac{1}{T} \max_{p_{\mathbf{X}}(\mathbf{X}): \frac{1}{T}\mathsf{E}[\mathbf{x}_k^\text{H} \mathbf{x}_k] \leq \alpha_k P, k\in \mathcal{I}_K} I(\mathbf{X};\mathbf{Y}),
	\end{equation}
	where the signal-to-noise ratio (SNR) is defined by $\mathsf{SNR} = \frac{P}{\sigma^2}$ and the maximization is taken over the distribution of $\mathbf{X}$, denoted by $p_{\mathbf{X}}(\mathbf{X})$, subject to the power constraint in \eqref{equ:power_constraint}.
	
\section{Capacity Bounds}	
The system capacity in \eqref{equ:capacity} is a very difficult problem and the exact solution is still unknown, especially in the circumstance with channel sparsity. In this section, we present upper and lower bounds to describe the system capacity. As seen later, the main contribution of this paper is to propose a blind detection scheme that provides a tight lower bound for the system capacity.

\subsection{Capacity Upper Bound}
	We start with a capacity upper bound. Ideally, when the CSI is perfectly known at the receiver, the capacity of the channel \eqref{equ:System1} is given by the following theorem \cite{Telatar99}.
	
	\begin{lemma}\label{lemma1}
		Assume that the channel matrix $\mathbf{H}$ is known at the receive terminal. Then, the channel capacity of the system in \eqref{equ:System1} is given by 
		\begin{equation} \label{equ:coh_capa}
			C_{\text{ideal}}(\mathsf{SNR}) = \mathsf{E} \left[\log\det(\mathbf{I}_K + \mathsf{SNR}\cdot  \mathbf{\Lambda} \mathbf{H}^\text{\emph{H}} \mathbf{H} ) \right],
		\end{equation}
		where $\mathbf{\Lambda} = \text{diag}\{\pmb{\alpha}\}$ with $ \pmb{\alpha} = [\alpha_1, \alpha_2, \cdots, \alpha_K]^\text{\emph{T}}$. The corresponding $\mathsf{DoF}$ is given by 
		\begin{equation} \label{equ:coh_DoF}
			\mathsf{DoF}_{\text{ideal}} =  \lim_{\mathsf{SNR}\rightarrow \infty} \frac{C_{\text{ideal}}(\mathsf{SNR})}{\log \mathsf{(SNR)}} = K.
		\end{equation}
	\end{lemma}
	Lemma \ref{lemma1} gives a performance upper bound for the considered massive MIMO system. We will show that, with channel sparsity, the ideal capacity can be closely approached, especially in the high SNR regime. 
	
\subsection{Capacity Lower Bounds}
	We now introduce lower bounds to the system capacity in \eqref{equ:capacity}. In general, every realizable detection scheme for the system in \eqref{equ:System1} provides a lower bound to the system capacity in \eqref{equ:capacity}. As aforementioned, all the existing schemes for massive MIMO perform very far away from the ideal capacity given in Lemma \ref{lemma1}, or in other words, the lower bounds provided by these schemes are very loose.
	
	In this paper, we propose a blind detection scheme to directly estimate $\mathbf{H}$ and $\mathbf{X}$ from the observed signal matrix $\mathbf{Y}$ following the maximum \emph{a posteriori} probability (MAP) principle. From the probability theory, the joint posterior probability density of $\mathbf{H}$ and $\mathbf{X}$ given $\mathbf{Y}$ is given by  
	\begin{align}
		p_{\mathbf{H},\mathbf{X}|\mathbf{Y}}(\mathbf{H},\mathbf{X}|\mathbf{Y})
		& \overset{(a)}{=} \frac{1}{p_{\mathbf{Y}} (\mathbf{Y})} p_{\mathbf{Y}| \mathbf{H},\mathbf{X}}(\mathbf{Y}|\mathbf{H},\mathbf{X}) p_{\mathbf{X}} (\mathbf{X}) p_{\mathbf{H}} (\mathbf{H}) \nonumber\\
		& \overset{(b)}{\propto} p_{\mathbf{Y}| \mathbf{H},\mathbf{X}}(\mathbf{Y}|\mathbf{H}\mathbf{X}) p_{\mathbf{X}} (\mathbf{X}) p_{\mathbf{H}} (\mathbf{H})\nonumber\\
		& \overset{(c)}{=} p_{\mathbf{W}}(\mathbf{Y}-\mathbf{H}\mathbf{X}) p_{\mathbf{X}} (\mathbf{X}) p_{\mathbf{H}} (\mathbf{H})\nonumber\\
		& \overset{(d)}{\propto} \exp\Big(-\frac{1}{\sigma^2}\| \mathbf{Y}-\mathbf{H}\mathbf{X} \|_2^2\Big) p_{\mathbf{X}} (\mathbf{X}) p_{\mathbf{H}} (\mathbf{H})
	\end{align}
	where step $(a)$ follows from the Bayes' rule and the fact that $\mathbf{H}$ and $\mathbf{X}$ are independent; the notation $\propto$ in step $(b)$ denotes equality up to a constant scaling factor; $p_{\mathbf{W}}(\cdot)$ in step $(c)$ denotes the probability density function (PDF) of noise $\mathbf{W}$; step $(d)$ follows from the fact that the elements of $\mathbf{W}$ are independently drawn from $\mathcal{CN}(0,\sigma^2)$. Then, the MAP estimates of $\mathbf{H}$ and $\mathbf{X}$, denoted respectively by $\hat{\mathbf{H}}$ and $\hat{\mathbf{X}}$, are given by 
	\begin{align}
		(\hat{\mathbf{H}}, \hat{\mathbf{X}}) &= \arg\max_{\mathbf{H}, \mathbf{X}: \frac{1}{T}\mathsf{E}[\mathbf{x}_k^\text{H} \mathbf{x}_k] \leq \alpha_k P, k\in \mathcal{I}_K } p_{\mathbf{H},\mathbf{X}|\mathbf{Y}}(\mathbf{H},\mathbf{X}|\mathbf{Y}) \nonumber\\
		&= \arg\max_{\mathbf{H}, \mathbf{X}: \frac{1}{T}\mathsf{E}[\mathbf{x}_k^\text{H} \mathbf{x}_k] \leq \alpha_k P, k\in \mathcal{I}_K } \exp\Big(-\frac{1}{\sigma^2}\| \mathbf{Y}-\mathbf{H}\mathbf{X} \|_2^2\Big) p_{\mathbf{X}} (\mathbf{X}) p_{\mathbf{H}} (\mathbf{H}).\label{equ:problem2}
	\end{align}
	The mutual information between $\hat{\mathbf{X}}$ and $\mathbf{X}$, denoted by $I(\hat{\mathbf{X}}; \mathbf{X})$, provides a lower bound to the capacity in \eqref{equ:capacity}. The corresponding DoF lower bound is given by 
	\begin{equation}\label{equ:dof2}
		\mathsf{DoF} =  \lim_{\mathsf{SNR}\rightarrow \infty} \frac{I(\hat{\mathbf{X}}; \mathbf{X})}{\log \mathsf{(SNR)}}.
	\end{equation}
	The problem in \eqref{equ:problem2} is in general difficult to solve, since the observed signal $\mathbf{Y}$ is a bilinear function of $\mathbf{H}$ and $\mathbf{X}$. Later, we will present a message-passing based algorithm to approximately solve \eqref{equ:problem2}, and will show that the lower bound provided by \eqref{equ:problem2} is much tighter than the existing bounds.
	
	To reveal the fundamental impact of channel sparsity on the DoF of a massive MIMO channel, we simplify \eqref{equ:problem2} by ignoring $p_{\mathbf{X}} (\mathbf{X})$ (which is irrelevant to the channel sparsity) and replacing $p_{\mathbf{H}} (\mathbf{H})$ with a Laplace distribution $\exp\Big(-\frac{\lambda}{\sigma^2}\|\mathbf{H}\|_1 \Big)$ (by following \cite{Donoho10,Sun17,SLearning}), yielding\footnote{In \eqref{equ:problem3}, the $\ell_1$-norm of $\mathbf{H}$ is used as the penalty term in the optimization problem. This is widely used in the area of machine learning to find a sparse solution \cite{SLearning}. Note that $\lambda \|\mathbf{H}\|_1$ in \eqref{equ:problem3b} implies a Laplace prior of $\mathbf{H}$. This explains why we choose a Laplace distribution for $p_{\mathbf{H}} (\mathbf{H})$ in \eqref{equ:problem3a}. }
	\begin{subequations}\label{equ:problem3}
	\begin{align}
		(\hat{\mathbf{H}}, \hat{\mathbf{X}}) &= \arg\max_{\mathbf{H}, \mathbf{X}: \frac{1}{T}\mathsf{E}[\mathbf{x}_k^\text{H} \mathbf{x}_k] \leq \alpha_k P, k\in \mathcal{I}_K } \exp\Big(-\frac{1}{\sigma^2}\| \mathbf{Y}-\mathbf{H}\mathbf{X} \|_2^2\Big) \exp\Big(-\frac{\lambda}{\sigma^2}\|\mathbf{H}\|_1 \Big)\label{equ:problem3a}\\
		&= \arg\min_{\mathbf{H}, \mathbf{X}: \frac{1}{T}\mathsf{E}[\mathbf{x}_k^\text{H} \mathbf{x}_k] \leq \alpha_k P, k\in \mathcal{I}_K } \| \mathbf{Y}-\mathbf{H}\mathbf{X} \|_2^2 + \lambda \|\mathbf{H}\|_1 \label{equ:problem3b}
	\end{align}	
	\end{subequations}
	where $\lambda$ is a regularization parameter that controls the tradeoff between the channel sparsity and the signal detection quality. Note that $l_1$-norm minimization is a common practice to deal with signal sparsity in the area of sparse signal recovery \cite{Sun17,Spielman12,Candes09,Donoho10}.
	
	Similarly to \eqref{equ:problem2}, $I(\hat{\mathbf{X}}; \mathbf{X})$ with $\hat{\mathbf{X}}$ given by \eqref{equ:problem3} gives a lower bound to the capacity in \eqref{equ:capacity} with corresponding achievable DoF given by \eqref{equ:dof2}. In the next section, we will analyze the DoF of the massive MIMO system based on the simplified problem in \eqref{equ:problem3}.
	
\subsection{Ambiguities in Blind Decection}\label{sec:ambiguity}
	Note that whether a detection method is \emph{blind} solely depends on the choice of the prior distribution $p_{\mathbf{X}}(\mathbf{X})$. For example, if a portion of the entries of $\mathbf{X}$ are perfectly known to the receiver, these entries are called pilot symbols and the corresponding scheme is training-based. Otherwise, $\mathbf{X}$ is said to be partially known to the receiver if $\mathbf{X}$ follows a certain distribution with a non-zero mean. In this paper, we focus on blind detection, in which $\mathbf{X}$ is unknown to the receiver. In other words, $\mathbf{X}$ is a random matrix following a zero-mean distribution.
		
	The blind detection problem formulated in \eqref{equ:problem3} suffers from ambiguities in signal estimation. Specifically, the objective function in \eqref{equ:problem3} is invariant to phase shifts and permutations of the rows of $\mathbf{X}$. Denote by $\mathbf{\Sigma}$ a diagonal matrix with phase shifts in the diagonal and by $\mathbf{\Pi}$ a permutation matrix. Then, the ambiguities are caused by the fact that if $(\hat{\mathbf{H}}, \hat{\mathbf{X}})$ is a solution to \eqref{equ:problem3}, then $(\hat{\mathbf{H}}^\prime = \hat{\mathbf{H}} \mathbf{\Pi}^{-1} \mathbf{\Sigma}^{-1}, \hat{\mathbf{X}}^\prime = \mathbf{\Sigma} \mathbf{\Pi} \hat{\mathbf{X}})$ is also a valid solution to \eqref{equ:problem3}. These ambiguities need to be appropriately handled in the transceiver design, as detailed in the subsequent sections.
	
	Besides the above ambiguities, there is another ambiguity inherent to blind detection. Denote by $\mathcal{S}_k$ and $\mathcal{S}_{k^\prime}$ the supports of the non-zero elements of $\mathbf{h}_k$ and $\mathbf{h}_{k^\prime}$, respectively. We refer to $\mathcal{S}_k$ as the \emph{sparsity pattern} of transmitter $k$. We now describe the sparsity-pattern collision problem by assuming $\mathcal{S}_k = \mathcal{S}_{k^\prime}$. Denote by $\mathbf{H}_{\{k, k^\prime\}} = [\mathbf{h}_k, \mathbf{h}_{k^\prime}] $ and by $\mathbf{X}_{\{k, k^\prime\}} = [\mathbf{x}_k, \mathbf{x}_{k^\prime}]^\text{T} $ the channel matrix and the signal matrix of transmitters $k$ and $k^\prime$, respectively. Then, the noise-free received signal, denoted by $\mathbf{Y}_{\{k, k^\prime\}}$, is given by
	\begin{equation}
		\mathbf{Y}_{\{k, k^\prime\}} = \mathbf{H}_{\{k, k^\prime\}} \mathbf{X}_{\{k, k^\prime\}} = \begin{bmatrix}
									\mathbf{h}_k & \mathbf{h}_{k^\prime}
		\end{bmatrix}
		\begin{bmatrix}
			\mathbf{x}_k^\text{T}\\ \mathbf{x}_{k^\prime}^\text{T}
		\end{bmatrix}.
	\end{equation}
	We construct an alternative factorization of $\mathbf{Y}_{\{k, k^\prime\}}$ as
	\begin{equation}\label{equ:Y_2}
		\mathbf{Y}_{\{k, k^\prime\}} = \tilde{\mathbf{H}}_{\{k, k^\prime\}}\tilde{\mathbf{X}}_{\{k, k^\prime\}},
	\end{equation}
	where
	\begin{subequations}
		\begin{align}
			\tilde{\mathbf{H}}_{\{k, k^\prime\}} &= \frac{1}{a+b-1} \begin{bmatrix}
							a\mathbf{h}_k+(1-b)\mathbf{h}_{k^\prime} & (1-a)\mathbf{h}_k+b\mathbf{h}_{k^\prime}
						\end{bmatrix}\\
			\tilde{\mathbf{X}}_{\{k, k^\prime\}} &= 					\begin{bmatrix}
							b\mathbf{x}_k^\text{T}-(1-a)\mathbf{x}_{k^\prime}^\text{T} \\ -(1-b)\mathbf{x}_k^\text{T}+a\mathbf{x}_{k^\prime}^\text{T}
						\end{bmatrix}
		\end{align}
	\end{subequations}
	and $a$ and $b$ are numbers with $a+b\neq 1$. Since the two users have a common support, i.e., $\mathcal{S}_k = \mathcal{S}_{k^\prime}$, we have $\|\tilde{\mathbf{H}}_{\{k, k^\prime\}}\|_1 = \|\mathbf{H}_{\{k, k^\prime\}}\|_1$ in general. Together with \eqref{equ:Y_2}, we conclude that the solution of \eqref{equ:problem3} is not unique in the sense that if $\mathbf{H}_{\{k, k^\prime\}}$ and $\mathbf{X}_{\{k, k^\prime\}}$ is a solution to \eqref{equ:problem3}, then $\tilde{\mathbf{H}}_{\{k, k^\prime\}}$ and $\tilde{\mathbf{X}}_{\{k, k^\prime\}}$ is also a solution to \eqref{equ:problem3}. This problem is referred to as sparsity-pattern collision. In practice, sparsity-pattern collision arises when the sparsity patterns of two or more users are close to each other, resulting in a non-zero probability of detection failure. We will see in the next section that, the failure probability can be made arbitrarily small in the high SNR regime if $N$ is sufficiently large and the channel is sufficiently sparse.
	
\section{DoF Analysis}	
\subsection{Heuristics}
In this section, we analyze the achievable DoF of the proposed blind detection scheme for the massive MIMO system. As aforementioned, the DoF of the massive MIMO system without exploiting the channel sparsity is given by $K(1-\frac{K}{T})$ \cite{Hassibi03,Zheng02}. Intuitively, each receive antenna receives a signal mixture from $K$ transmit terminals, and therefore needs to use $K$ time slots to identify $K$ channel coefficients, which contributes to the fractional DoF gap of $\frac{K}{T}$ from the ideal DoF. The compressed-sensing based training schemes \cite{Bajwa10} can reduce the fractional DoF gap to $c\frac{K}{T}$, where $c$ is a constant coefficient between zero and one determined by the channel sparsity level $\rho$. In the following, we will show that the proposed blind detection scheme can further reduce the fractional DoF gap to $\frac{1}{T}$, which brings forward a significant step towards the understanding of the fundamental capacity of the massive MIMO system with channel sparsity.

Our analysis is based on the problem formulation in \eqref{equ:problem3}. The DoF characterizes the behaviour of the minimizer $(\hat{\mathbf{H}}, \hat{\mathbf{X}})$ when the SNR goes to infinity, or equivalently, when the noise level $\sigma^2$ goes to zero. This means that we need to understand how a small perturbation of the additive noise $\mathbf{W}$ affects the minimizer $\hat{\mathbf{H}}$ and $\hat{\mathbf{X}}$, and how the system parameters, such as $N$, $K$, $T$, and $\rho$, interact with each other to guarantee the existence of the minimizer $(\hat{\mathbf{H}}, \hat{\mathbf{X}})$ around the ground truth of $(\mathbf{H}, \mathbf{X})$. Moreover, in the analysis, we need to appropriately handle the phase and permutation ambiguities as well as the sparsity pattern collision problem described in Section III-C.
	
\subsection{Assumptions}
	We start with some regularity conditions.
	
	\textbf{Assumption 1} (Coefficient independence): 
	\begin{equation}
		\mathsf{E}[h_{n,k} h_{m,l}^*] = 0, \text{ for any }(n,k), (m,l) \in \mathcal{S} \text{ with } n\neq m \text{ or } k\neq l. 
	\end{equation}
	
	\textbf{Assumption 2} (Coefficient boundedness):
	\begin{subequations}\label{assumption2}
		\begin{align}
			\mathsf{P}(|h_{n,k}| > \epsilon ) &= 1, \text{ for some } \epsilon >0 \text{ and } (n,k) \in \mathcal{S} \label{assumption2_a}\\
			\text{ and } \qquad\qquad\qquad
			\mathsf{P}(\|\mathbf{h}_k\|_2< M_h ) &= 1, \text{ for some } M_h \text{ and } k\in \mathcal{I}_K \qquad\qquad\qquad\label{assumption2_b}
		\end{align}
	\end{subequations}
	where $\mathsf{P}(\cdot)$ represents the probability function. 
	
	The assumption in \eqref{assumption2_a} is reasonable since in practical systems the channel coefficients with $|h_{n,k}| <\epsilon$ can be ignored without compromising the system performance, provided that $\epsilon$ is sufficiently small. As for the assumption in \eqref{assumption2_b}, the length of the channel vector $\mathbf{h}_k$ for each transmitter $k$ is uniformly bounded since the received signal power from each transmitter is always bounded due to channel attenuation. 
	
	\textbf{Assumption 3} (Noise boundedness):
	\begin{equation}\label{equ:Zupper}
		\mathsf{P}\left(\|\mathbf{w}_n \|_2 > M_w \right) = 0, \text{ for some $M_w$ and } n\in \mathcal{I}_{N}
	\end{equation}
	where $\mathbf{w}_n$ denotes the transpose of the $n$th row of $\mathbf{W}$. 
	
	This assumption can be approximately satisfied by the Gaussian noise drawn from $\mathcal{CN}(0,\sigma^2)$. For example, for $M_w = \sqrt{3T\sigma^2}$, $\mathsf{P}\left(\|\mathbf{w}_n \|_2 > M_w \right)$ is only $0.00135$. In fact, with an appropriately chosen $M_w$, $\mathsf{P}\left(\|\mathbf{w}_n \|_2 > M_w \right)$ can be made arbitrarily small and thus can be ignored from a practical point of view. 
		
\subsection{Main Result}
	The main results of the paper are presented here and their proofs are given in the next subsection.
	
	\begin{theorem}\label{theorem:DoF}
		Assume $\rho \lesssim \frac{1}{K} \sqrt{\frac{T}{\log K}}$ and $K \lesssim T^2$. Then, for any $\eta > 0$ and $N \gtrsim TK^3+\eta K^2 $, the DoF of the system given in \eqref{equ:System1} is lower bounded by
		\begin{equation}\label{equ:BlindDoF}
			\mathsf{DoF}_{\text{\emph{blind}}} = (1-e^{-\eta}) K\left(1-\frac{1}{T} \right).
		\end{equation}
	\end{theorem}
	
	\emph{Remark 1}: In \eqref{equ:BlindDoF}, $e^{-\eta}$ can be understood as the detection failure probability. The sparsity-pattern collision discussed in Section \ref{sec:ambiguity} is one factor to cause detection failure. Intuitively, in practical massive MIMO systems, the sparsity-pattern collision arises when users are geographically close to each other. Signals from co-located users undergo similar scattering, leading to similar AoAs at the receive antenna array and hence similar sparsity patterns in the angular domain. Increasing the number of antennas improves the AoA resolution of an antenna array, and therefore reduces the probability of sparsity-pattern collision as well as the detection failure probability $e^{-\eta}$. This explains why in Theorem \ref{theorem:DoF} a larger $N$ allows a larger $\eta$ (and hence a smaller detection failure probability $e^{-\eta}$). The DoF lower bound in \eqref{equ:BlindDoF} can approach $K\left(1-\frac{1}{T}\right)$ arbitrarily closely provided that $N$ is sufficiently large.
	
	\emph{Remark 2}: From Theorem \ref{theorem:DoF}, our scheme achieves a DoF arbitrarily close to $K(1-1/T)$, provided that $N$ is sufficiently large and the channel is sufficiently sparse. Compared with the ideal case in \eqref{equ:coh_DoF}, our scheme has a fractional DoF loss of $1/T$. This loss is caused by the phase ambiguity of $\mathbf{H}$ and $\mathbf{X}$ discussed in Section \ref{sec:ambiguity}. In other words, a fraction $1/T$ of DoF is required to eliminate the ambiguity. We emphasize that the fractional DoF loss $1/T$ is independent of $K$, which is the key advantage of the proposed blind detection scheme compared with other existing counterparts.    
	
\subsection{Proof of Theorem \ref{theorem:DoF}}
	For simplicity, we assume that $\alpha_1 = \alpha_2 = \cdots = \alpha_K = 1/K$, and that the elements of $\mathbf{X}$ are independently drawn from $\mathcal{CN}(0,P/K)$. The optimization of the power coefficients $\{\alpha_k \}$ and the distribution of $\mathbf{X}$ may lead to a better performance, but is out of the scope of this paper. 
	
	We focus on the lower bound in \eqref{equ:problem3}. From \eqref{equ:problem3}, we see that, if $(\hat{\mathbf{H}}, \hat{\mathbf{X}})$ is a solution to \eqref{equ:problem3}, then any $(\hat{\mathbf{H}} \mathbf{\Pi}^{-1} \mathbf{\Sigma}^{-1}, \mathbf{\Sigma} \mathbf{\Pi} \hat{\mathbf{X}})$ is also a valid solution to \eqref{equ:problem3} with $\mathbf{\Sigma}$ being a diagonal matrix for phase ambiguity and $\mathbf{\Pi}$ being a permutation matrix for permutation ambiguity. By considering these inherent ambiguities, an estimate $\hat{\mathbf{X}}$ for \eqref{equ:problem3} can be modeled by
	\begin{equation}\label{equ:x_hat_model1}
		\hat{\mathbf{X}} = \mathbf{\Sigma} \mathbf{\Pi}  (\mathbf{X}  + \mathbf{\Delta}),
	\end{equation}
	where the estimation error $\mathbf{\Delta}$ is zero-mean and uncorrelated with $\mathbf{X}$. Note that $\mathbf{\Sigma}$ and $\mathbf{\Pi}$ are regarded as deterministic parameters unknown to the receiver. We have the following result. 
	
	\begin{lemma}\label{lemma2}
		Assume that $\rho \lesssim \frac{1}{K} \sqrt{\frac{T}{\log K}}$ and $K\lesssim T^2$. Then, in the high SNR regime, for any $\eta >0$ and $\delta >0 $, if $N \gtrsim (TK^3+\eta K^2) \left(\frac{1}{\delta}\cdot \frac{M_w}{M_h}\right)^2$, there admits a solution $\hat{\mathbf{X}}$ for the problem in \eqref{equ:problem3} such that $\|\mathbf{\Delta} \|_F^2  < \delta^2$ in a probability at least $1-e^{-\eta}$.
	\end{lemma}
	\begin{proof}
		The proof of Lemma \ref{lemma2} is mainly based on the result of [36, Theorem 2]. Recall that the elements of $\mathbf{X}$ are independently drawn from $\mathcal{CN}(0,P/K)$. Then, $\mathbf{X}$ satisfies the conditions [36, eqn. (17)] and [36, eqn. (18)] provided $\rho \lesssim \frac{1}{K} \sqrt{\frac{T}{\log K}}$. With Assumptions 1-3, condition [36, eqn. (19)] is satisfied, provided $K\lesssim T^2$. In the high SNR regime, as the noise level $\sigma^2$ tends to zero, we can always find appropriate parameters $\lambda$ and $\delta$ to satisfy the conditions [36, eqn. (22)-(25)]. Therefore, from [36, Theorem 2], the problem in \eqref{equ:problem3} admits a local minimum within radius $\delta$ centered around $\mathbf{X}$ (i.e., $\|\mathbf{\Delta} \|_F^2  < \delta^2$) with probability at least $1-e^{-\eta}$, provided that $N \gtrsim (TK^3+\eta K^2) \left(\frac{1}{\delta}\cdot \frac{M_w}{M_h}\right)^2$. This completes the proof of Lemma \ref{lemma2}.
	\end{proof}

	From Lemma \ref{lemma2}, we have $N \gtrsim (TK^3+\eta K^2) \left(\frac{1}{\delta}\cdot \frac{M_w}{M_h}\right)^2 $ and $\|\mathbf{\Delta} \|_F^2 < \delta^2$, yielding 
	\begin{equation}\label{equ:Delta_F}
		\|\mathbf{\Delta} \|_F^2 \lesssim \frac{c_1 (TK^3+\eta K^2)}{N} \frac{M_w^2}{M_h^2}.
	\end{equation}
	
	We now describe how to determine the diagonal matrix $\mathbf{\Sigma}$ in \eqref{equ:x_hat_model1}. We use one symbol in each $\mathbf{x}_k$ to estimate $\mathbf{\Sigma}$.\footnote{For ease of analysis, we assume coherent detection for each $\mathbf{x}_k$ in the proof. That is, one symbol in each $\mathbf{x}_k$ is used to estimate $\Sigma_k$, and the estimated $\Sigma_k$ is then used to detect the rest of $\mathbf{x}_k$. We emphasize that non-coherent detection can also be used here. For example, differential coding can be used to remove the phase uncertainty caused by $\Sigma_k$. We conjecture that coherent and non-coherent detections achieve the same DoF, as motivated by the result in \cite{Zheng02}.} Without loss of generality, we assume that $x_{k,1} = \sqrt{\frac{P}{K}}, k \in \mathcal{I}_K $ is known to the receiver, where $x_{k,1}$ is the first entry of $\mathbf{x}_k$. Denote by $\pmb{\delta}_k$ the transpose of the $k$th row of $\mathbf{\Delta}$. Then, from \eqref{equ:x_hat_model1}, we obtain
		\begin{equation}
			\hat{x}_{k,1} = \Sigma_k \Big(\sqrt{\frac{P}{K}} + \delta_{\pi(k),1}\Big),
		\end{equation}
		where $\hat{x}_{k,1}$ denotes the $(k,1)$th element of $\hat{\mathbf{X}}$, $\Sigma_k$ represents the $k$th diagonal element of $\mathbf{\Sigma}$, $\pi(\cdot)$ is the permutation function corresponding to $\mathbf{\Pi}$, and $\delta_{\pi(k),1}$ is the first entry of $\pmb{\delta}_{\pi(k)}$. Then, an estimate of $\Sigma_k$, denoted by $\hat{\Sigma}_k$, is given by 
		\begin{equation}\label{equ:SigmaEst}
			\hat{\Sigma}_k = \sqrt{\frac{K}{P}} \hat{x}_{k,1} = \Sigma_k + \sqrt{\frac{K}{P}} \Sigma_k \delta_{\pi(k),1}.
		\end{equation} 
		The corresponding estimation error is given by
		\begin{equation} \label{equ:SigmaError}
			\Sigma_k - \hat{\Sigma}_k = -\sqrt{\frac{K}{P}} \Sigma_k \delta_{\pi(k),1}.
		\end{equation}
		
		Denote by $\tilde{\mathbf{x}}_k \triangleq [x_{k,2},\cdots, x_{k,T}]^\text{T}\in \mathbb{C}^{(T-1)\times 1}$, $\hat{\tilde{\mathbf{x}}}_k \triangleq [\hat{x}_{k,2},\cdots, \hat{x}_{k,T}]^\text{T}\in \mathbb{C}^{(T-1)\times 1}$, and $\tilde{\pmb{\delta}}_k = [\delta_{k,2}, \cdots, \delta_{k,T}]^\text{T}\in \mathbb{C}^{(T-1)\times 1}$. From \eqref{equ:x_hat_model1}, we have
		\begin{equation}
			\hat{\tilde{\mathbf{x}}}_k = \hat{\Sigma}_k \tilde{\mathbf{x}}_{\pi(k)} + (\Sigma_k -\hat{\Sigma}_k) \tilde{\mathbf{x}}_{\pi(k)}  + \Sigma_k \tilde{\pmb{\delta}}_{\pi(k)}
							 = \hat{\Sigma}_k \tilde{\mathbf{x}}_{\pi(k)}  + \mathbf{v}_k, \label{equ:x_hat3}
		\end{equation}
		where $\mathbf{v}_k \triangleq (\Sigma_k -\hat{\Sigma}_k) \tilde{\mathbf{x}}_{\pi(k)}  + \Sigma_k \tilde{\pmb{\delta}}_{\pi(k)}$. The average variance of the entries of $\mathbf{v}_k$ is given by
		\begin{align}
			\sigma_{v_k}^2 = \frac{1}{T-1} \mathsf{E}[\mathbf{v}_k^\text{H} \mathbf{v}_k] 
			&= \frac{1}{T-1} \mathsf{E}\left[\big((\Sigma_k -\hat{\Sigma}_k) \tilde{\mathbf{x}}_{\pi(k)}  + \Sigma_k \tilde{\pmb{\delta}}_{\pi(k)}\big)^\text{H} \big((\Sigma_k -\hat{\Sigma}_k) \tilde{\mathbf{x}}_{\pi(k)}  + \Sigma_k \tilde{\pmb{\delta}}_{\pi(k)}\big)\right] \nonumber\\
			&\overset{(a)}{=} \frac{1}{T-1}\mathsf{E}\left[ |\Sigma_k -\hat{\Sigma}_k|^2 \tilde{\mathbf{x}}_{\pi(k)}^\text{H} \tilde{\mathbf{x}}_{\pi(k)}\right] + \frac{1}{T-1}\mathsf{E}\left[|\Sigma_k|^2 \tilde{\pmb{\delta}}_{\pi(k)}^\text{H} \tilde{\pmb{\delta}}_{\pi(k)}\right]\nonumber\\
			&\overset{(b)}{=} \frac{P}{K} \mathsf{E}\left[|\Sigma_k -\hat{\Sigma}_k|^2\right] + \sigma_{\delta_{\pi(k)}}^2 |\Sigma_k |^2 \nonumber\\
			&\overset{(c)}{=} \sigma_{\delta_{\pi(k)}}^2 |\Sigma_k |^2 + \sigma_{\delta_{\pi(k)}}^2 |\Sigma_k |^2 \nonumber\\
			&= 2\sigma_{\delta_{\pi(k)}}^2 |\Sigma_k |^2 \label{equ:Rv}
		\end{align}
		where $\sigma_{\delta_{k}}^2$ is the average variance of the entries of $\pmb{ \delta}_{k}$, defined by
		\begin{equation} \label{equ:sigma_delta_k}
			\sigma_{\delta_{k}}^2 \triangleq \frac{1}{T} \mathsf{E}[\pmb{ \delta}_{k}^\text{H} \pmb{ \delta}_{k}].
		\end{equation}
		In \eqref{equ:Rv}, step $(a)$ follows from the fact that $\mathbf{X}$ and $\mathbf{\Delta}$ are zero-mean and uncorrelated; step $(b)$ follows from \eqref{equ:sigma_delta_k}; step $(c)$ follows from \eqref{equ:SigmaError}. In general, $\mathbf{v}_k$ is correlated with the signal $\hat{\Sigma}_k \tilde{\mathbf{x}}_{\pi(k)}$, and is not necessarily Gaussian. From \cite{Medard00}, the mutual information between $\hat{\tilde{\mathbf{x}}}_k$ and $\tilde{\mathbf{x}}_k$ is lower bounded by the case when $\mathbf{v}_k$ is independent and Gaussian. Then
		\begin{align}
			I(\hat{\tilde{\mathbf{x}}}_k; \tilde{\mathbf{x}}_k) &\overset{(a)}{=} (T-1) \mathsf{E}\left[\log\left(1 + \frac{\frac{P}{K}|\Sigma_k |^2\left(1+\frac{K}{P}\sigma_{\delta_{\pi(k)}}^2 \right) }{\sigma_{v_k}^2} \right) \right] \nonumber\\
			&\overset{(b)}{=} (T-1) \mathsf{E} \left[\log \left(1+\frac{P}{2K\sigma_{\delta_{\pi(k)}}^2}\left(1+ \frac{K}{P}\sigma_{\delta_{\pi(k)}}^2 \right) \right) \right]\nonumber\\
			&\overset{(c)}{=} (T-1) \log\left( \frac{P}{\sigma_{\delta_{\pi(k)}}^2} \right) + o(T)
		\end{align}
		where step $(a)$ follows from \eqref{equ:SigmaEst} and \eqref{equ:x_hat3}; step $(b)$ follows from \eqref{equ:Rv}; and step $(c)$ follows from the high SNR approximation. Then, we have
		\begin{equation}\label{equ:sumRate1}
			\sum_{k=1}^K I(\hat{\tilde{\mathbf{x}}}_k; \tilde{\mathbf{x}}_k) = (T-1)\sum_{k=1}^K \log\left( \frac{P}{\sigma_{\delta_{\pi(k)}}^2} \right) + o(T) 
			\overset{(a)}{\geq} (T-1)K\log\left( \frac{P}{\frac{\mathsf{E}[\|\mathbf{\Delta}\|_F^2] }{KT} } \right) + o(T)
		\end{equation}
		where step (a) follows from $\prod_{k=1}^K \sigma_{\delta_{\pi(k)}}^2 \leq \left(\frac{\mathsf{E}[\|\mathbf{\Delta}\|_F^2] }{KT}\right)^K $ for $\sum_{k=1}^K \sigma_{\delta_{\pi(k)}}^2 = \frac{\mathsf{E}[\|\mathbf{\Delta}\|_F^2] }{T}$, with the equality holds for $\sigma_{\delta_{1}}^2 = \cdots = \sigma_{\delta_{K}}^2 = \frac{\mathsf{E}[\|\mathbf{\Delta}\|_F^2] }{KT}$. 
				
		We now consider the impact of permutation $\mathbf{\Pi}$ in \eqref{equ:x_hat_model1}. From \eqref{equ:x_hat3}, we see that the receiver decodes each $\tilde{\mathbf{x}}_{\pi(k)}$ from $\hat{\tilde{\mathbf{x}}}_{\pi(k)}$. Since $\mathbf{\Pi}$ is unknown, the receiver needs to associate each decoded codeword to the corresponding transmitter. This can be done by inserting a transmitter label into each codeword. To identify $K$ transmitters, each label costs $\lceil \log K \rceil$ bits. Thus, an achievable rate of the system in \eqref{equ:System1} is given by
		\begin{equation}
			\mathsf{R}_{\text{blind}} = \frac{1}{T} \sum_{k=1}^K I(\hat{\tilde{\mathbf{x}}}_k; \tilde{\mathbf{x}}_k) -\frac{K\lceil \log K \rceil}{T} 
			= \frac{T-1}{T}K\log\left( \frac{P}{\frac{\mathsf{E}[\|\mathbf{\Delta}\|_F^2] }{KT} } \right) +o(1) -\frac{K\lceil \log K \rceil}{T}. \label{equ:rate1}
		\end{equation}
		From \eqref{equ:Delta_F}, with a probability of at least $1-e^{-\eta }$, we obtain
		\begin{equation} \label{equ:sigma_delta}
			\frac{\|\mathbf{\Delta} \|_F^2}{KT} \lesssim \frac{(TK^2+\eta K)}{NT} \frac{M_w^2}{M_h^2}.
		\end{equation}
		Thus, we have
		\begin{equation}
			\mathsf{R}_{\text{blind}} \geq (1-e^{-\eta})K\left(1-\frac{1}{T}\right) \log \frac{PM_h^2}{M_w^2} + \tilde{c} + o(1), \label{equ:rate2}
		\end{equation}
		where 
		\begin{equation}
			\tilde{c} = (1-e^{-\eta})K\left(1-\frac{1}{T}\right)\log \frac{NT}{TK^2+\eta K}-\frac{K\lceil \log K \rceil}{T}. 
		\end{equation}
		Since $M_h^2$ is a constant and $M_w^2$ is proportional to $\sigma^2$ (from Assumption 3 and the discussions therein), the DoF of the proposed scheme is bounded by 
		\begin{equation}
			\mathsf{DoF}_{\text{blind}} = \lim_{\mathsf{SNR} \rightarrow \infty} \frac{\mathsf{R}_{\text{blind}}}{\log \mathsf{SNR}} = (1-e^{-\eta})K\left(1-\frac{1}{T} \right),
		\end{equation}
		which completes the proof of Theorem \ref{theorem:DoF}.
	
\section{Blind Detection Algorithm}\label{sec:algorithm}
\subsection{Preliminaries}
	In this section, we present a blind detection algorithm to realize the DoF advantage promised by our analysis in practical system settings. We focus on the problem in \eqref{equ:problem2}, which is in general difficult to solve optimally. As aforementioned, the blind detection problem under concern is very similar to the matrix factorization problem in sparse dictionary learning \cite{Sun17}. In the context of dictionary learning, the blind detection problem can be rephrased as to learn the dictionary $\mathbf{X}$ and the corresponding sparse representation $\mathbf{H}$. Dictionary learning algorithms, such as the K-SVD algorithm \cite{KSVD06}, the SPAMS algorithm \cite{SPAMS10}, and the BiG-AMP algorithm \cite{PSchniter14}, can be potentially applied to approximately solve \eqref{equ:problem2}. However, there is a notable difference in the system configuration between the massive MIMO system and a typical dictionary learning problem. In dictionary learning, the dictionary $\mathbf{X}$ is usually overcomplete, i.e., $T \leq K$, so as to allow more flexible dictionaries and richer data representation. However, in massive MIMO, the setting of $T > K$ is of more relevance since the coherence time $T$ is usually in hundreds and the number of users is usually in tens.
	
	It is known that BiG-AMP is a state-of-the-art dictionary learning algorithm. However, when applied to solve \eqref{equ:problem2}, BiG-AMP does not perform well in practical systems where $T > K$. To address this issue, we propose to preprocess the observed signal matrix $\mathbf{Y}$ by projecting it onto the signal space $\mathbf{X}$ and then apply BiG-AMP to the image of the projection. In addition, we describe how to eliminate the phase and permutation ambiguities. The resulting algorithm, referred to as P-BiG-AMP, is presented in the following subsections. Note that we always assume $T > K$ in the subsequent discussions. For the case of $T \leq K$, we simply skip the signal projection operation and directly apply BiG-AMP to solve \eqref{equ:problem2}.
	
\subsection{Signal Projection}
	From \eqref{equ:System1}, the noise-free received signal $\mathbf{HX}$ has the same row space as $\mathbf{X}$ does. Thus, the row space of $\mathbf{X}$ can be estimated based on $\mathbf{Y}$ as follows. 
	
	Let the singular value decomposition (SVD) of $\mathbf{Y}$ be 
	\begin{equation}
		\mathbf{Y} = \mathbf{UD}\mathbf{V}^\text{H},
	\end{equation}
	where $\mathbf{U}\in \mathbb{C}^{N\times N}$ and $\mathbf{V}\in \mathbb{C}^{T\times T}$ are unitary matrices, and $\mathbf{D}\in \mathbb{R}^{N\times T}$ is a diagonal matrix with diagonal elements arranged in a descend order. Then, we partition $\mathbf{V}$ as
	\begin{equation}
		\mathbf{V} = [\mathbf{V}_1, \mathbf{V}_2],
	\end{equation}
	where $\mathbf{V}_1 \in \mathbb{C}^{T\times K}$ corresponds to the largest $K$ singular values, and $\mathbf{V}_2 \in \mathbb{C}^{T\times (T-K)}$ corresponds to the other singular values. Note that the column space of $\mathbf{V}_1$ gives an estimate of the row space of $\mathbf{X}$.
	
	We then project the observed signal $\mathbf{Y}$ onto the column space of $\mathbf{V}_1$, yielding 
	\begin{equation}\label{equ:system5}
		\mathbf{Y}^\prime = \mathbf{Y} \mathbf{V}_1 = \mathbf{HX}\mathbf{V}_1 + \mathbf{W} \mathbf{V}_1 = \mathbf{H}\mathbf{X}^\prime + \mathbf{W}^\prime,
	\end{equation} 
	where $\mathbf{X}^\prime \triangleq \mathbf{X}\mathbf{V}_1 \in \mathbb{C}^{K\times K} $ and $\mathbf{W}^\prime \triangleq \mathbf{W}\mathbf{V}_1 \in \mathbb{C}^{N\times K} $. After projection, the average transmission power constraint for $\mathbf{X}^\prime$ is given by
	\begin{equation}
		\frac{1}{K}\mathsf{E}[(\mathbf{x}_k^{\prime})^{\text{H}} \mathbf{x}_k^\prime ] \leq \alpha_k P, \text{ for all } k \in \mathcal{I}_K \triangleq \{1,2,\cdots, K\},
	\end{equation}
	where $\mathbf{x}_k^{\prime} \in \mathbb{C}^{K\times 1}$ is the transpose of the $k$th row of $\mathbf{X}^\prime$. 
	
	The next step is to factorize $\mathbf{Y}^\prime$ to obtain $\mathbf{H}$ and $\mathbf{X}^\prime$. Denote by $p_{\mathbf{X}^\prime}(\mathbf{X}^\prime)$ the prior distribution of $\mathbf{X}^\prime$. Then, similar to \eqref{equ:problem2}, the blind detection problem can be written as
	\begin{equation}\label{equ:problem4}
		(\hat{\mathbf{H}}, \hat{\mathbf{X}}^\prime) = \arg\max_{\mathbf{H}, \mathbf{X}^\prime: \frac{1}{K}\mathsf{E}[(\mathbf{x}_k^\prime)^\text{H} \mathbf{x}_k^\prime] \leq \alpha_k P, k\in \mathcal{I}_K } \exp\Big(-\frac{1}{\sigma^2}\| \mathbf{Y}-\mathbf{H}\mathbf{X}^\prime \|_2^2\Big) p_{\mathbf{X}^\prime} (\mathbf{X}^\prime) p_{\mathbf{H}} (\mathbf{H}).
	\end{equation}
	The BiG-AMP algorithm can then be directly used to solve the problem in \eqref{equ:problem4}.

\subsection{BiG-AMP Algorithm}
	\begin{figure}[!t]
		\centering
		\includegraphics[width = 0.8\textwidth]{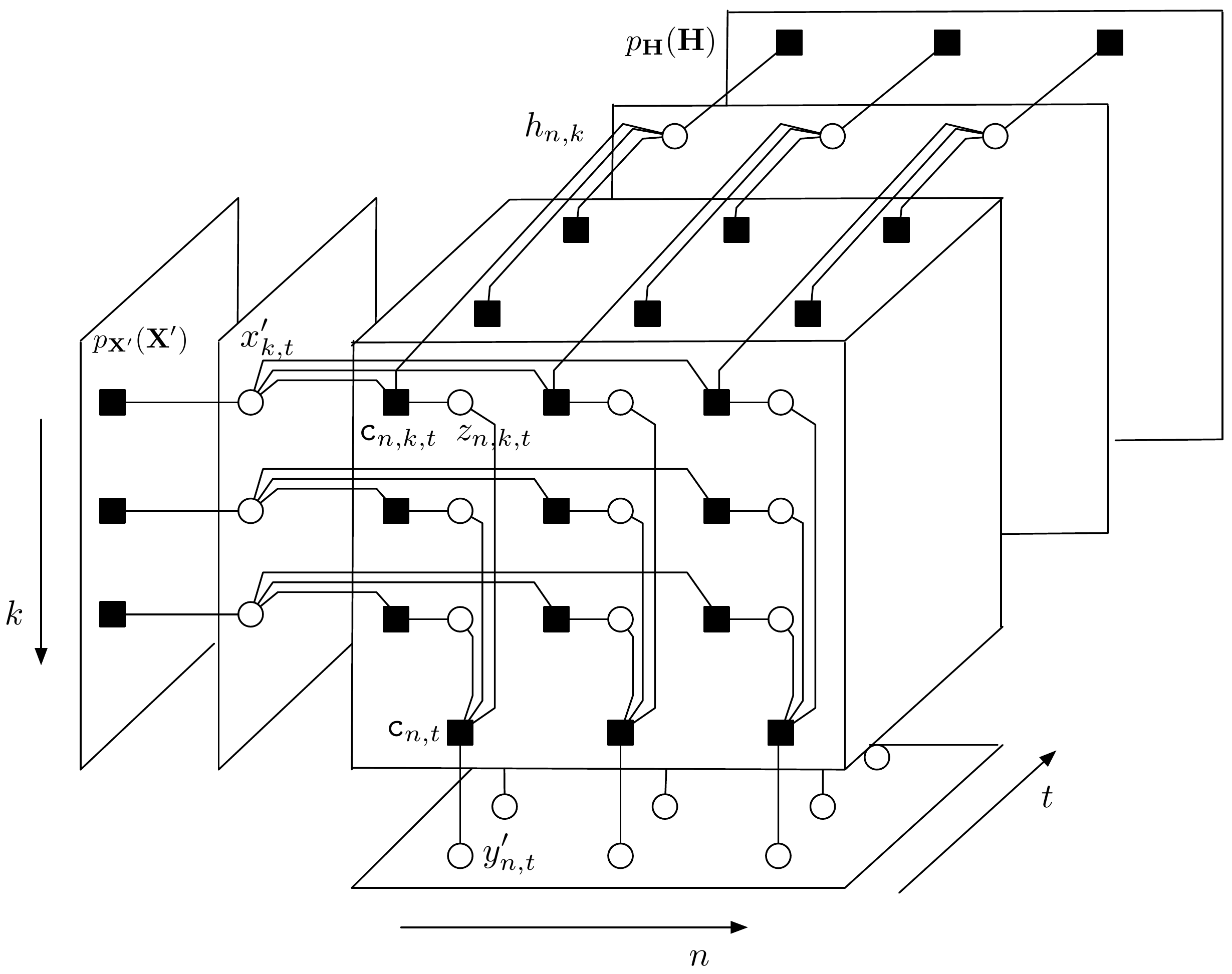}
	 	\caption{The factor graph representation for the problem in \eqref{equ:problem2}.}\label{Fig:FactorGraph}
	\end{figure}
	
	For completeness, we describe the BiG-AMP algorithm as follows. We start with the factor graph representation associated with the problem in \eqref{equ:problem4}. Denote by $y_{n,t}^\prime$ the $(n,t)$th element of $\mathbf{Y}^\prime$. Then, the system model \eqref{equ:system5} can be equivalently written as
	\begin{equation}
		y_{n,t}^\prime = \sum_{k = 1}^K z_{n,k,t} + w_{n,t}^\prime, \text{ for all } n\in \mathcal{I}_N \text{ and } t\in \mathcal{I}_K \label{equ:sum_constraint}
	\end{equation}
	where
	\begin{equation}
		z_{n,k,t} = h_{n,k} x_{k,t}^\prime, \text{ for all } k \in \mathcal{I}_K \label{equ:prod_constraint}
	\end{equation}
	with $x_{k,t}^\prime$ and $w_{n,t}^\prime$ are the $(k,t)$th and $(n,t)$th element of $\mathbf{X}^\prime$ and $\mathbf{W}^\prime$, respectively. With \eqref{equ:sum_constraint} and \eqref{equ:prod_constraint}, we construct a factor graph illustrated in Fig. \ref{Fig:FactorGraph}. The factor graph consists of two types of nodes: variable nodes and check nodes. Variable nodes include $\{h_{n,k}\}$, $\{x_{k,t}^\prime\}$, $\{z_{n,k,t}\}$, and $\{y_{n,t}^\prime \}$; check nodes include the $(n,t)$th equation in \eqref{equ:sum_constraint} (denoted by $\mathtt{c}_{n,t}$) for all $n\in \mathcal{I}_N$ and $t\in \mathcal{I}_K$, the $(n,k,t)$th equation in \eqref{equ:prod_constraint} (denoted by $\mathtt{c}_{n,k,t}$) for all $n\in \mathcal{I}_N$, $k\in \mathcal{I}_K$, and $t\in \mathcal{I}_K$, and also the prior distributions $p_{\mathbf{H}}(\mathbf{H})$ and $p_{\mathbf{X}^\prime}(\mathbf{X}^\prime)$. In Fig. \ref{Fig:FactorGraph}, all variable nodes appear as white circles and all check nodes appear as black boxes. There is an edge connection between a variable node and a check node when the variable node appears in the equation corresponding to the check node.
	
\begin{table}[!t]
	\caption{Algorithm 1: The BiG-AMP Algorithm}
	\label{tab:algorithm}
	\centering
	\begin{tabular}{l r}
	\hline
		\textbf{Input}: $\mathbf{Y}^\prime$, prior distributions $p_{\mathbf{H}}(\mathbf{H}) $ and $p_{\mathbf{X}^\prime}(\mathbf{X}^\prime)$.\\
		\textbf{Initialization}: $\hat{h}_{n,k}(1) = 0$, $v_{n,k}^h = 1$, $\hat{x}_{k,t}^\prime(1)$ is randomly drawn from $p_{x_{k,t}^\prime}(x_{k,t}^\prime)$,\\
		$v_{k,t}^x(1) = \alpha_k P$, and $\hat{s}_{n,t}(0) = 0$.\\
		\textbf{for} $m = 1, \cdots, M_{\text{max}}$\\
		\ \ \textbf{for} $l = 1, \cdots, L_{\text{max}}$\\
		\ \ \quad $\forall n, t: \bar{v}_{n,t}^p(l) = \sum_{k=1}^K |\hat{h}_{n,k}(l)|^2 v_{k,t}^x(l) + v_{n,k}^h(l) |\hat{x}_{k,t}^\prime(l)|^2$ & (A1)\\
		\ \ \quad $\forall n, t: \bar{p}_{n,t}(l) = \sum_{k=1}^K  \hat{h}_{n,k}(l)\hat{x}_{k,t}^\prime(l)$ & (A2)\\
		\ \ \quad $\forall n, t: v_{n,t}^p(l) = \bar{v}_{n,t}^p + \sum_{k=1}^K v_{n,k}^h(l) v_{k,t}^x(l)$ & (A3)\\
		\ \ \quad $\forall n, t: \hat{p}_{n,t}(l) = \bar{p}_{n,t}(l)  - \hat{s}_{n,t}(l-1) \bar{v}_{n,t}^p(l)$ & (A4)\\
		\ \ \quad $\forall n, t: v_{n,t}^z(l) = \frac{v_{n,t}^p(l)\sigma^2 }{v_{n,t}^p(l)+\sigma^2}$ & (A5)\\
		\ \ \quad $\forall n, t: \hat{z}_{n,t}(l) = \frac{v_{n,t}^p(l)}{v_{n,t}^p(l)+\sigma^2}(y_{n,t} - \hat{p}_{n,t}(l) ) + \hat{p}_{n,t}(l)$ & (A6)\\
		\ \ \quad $\forall n, t: v_{n,t}^s(l) = (1-v_{n,t}^z(l)/v_{n,t}^p(l))/v_{n,t}^p(l)$ & (A7)\\
		\ \ \quad $\forall n, t: \hat{s}_{n,t}(l) = (\hat{z}_{n,t}(l) - \hat{p}_{n,t}(l))/v_{n,t}^p(l)$ & (A8)\\
		\ \ \quad $\forall n, k: v_{n,k}^q(l) = \left(\sum_{t=1}^T |\hat{x}_{k,t}^\prime(l)|^2 v_{n,t}^s(l) \right)^{-1} $ & (A9) \\
		\ \ \quad $\forall n, k: \hat{q}_{n,k}(l) = \hat{h}_{n,k}(l)\Big(1-v_{n,k}^q(l) \sum_{t=1}^T v_{k,t}^x(l) v_{n,t}^s(l) \Big) + v_{n,k}^q(l) \sum_{t=1}^T (\hat{x}_{k,t}^\prime(l))^* \hat{s}_{n,t}(l) $ & (A10) \\
		\ \ \quad $\forall k, t: v_{k,t}^r(l) = \left(\sum_{n=1}^N |\hat{h}_{n,k}(l)|^2 v_{n,t}^s(l) \right)^{-1} $ & (A11) \\
		\ \ \quad $\forall k, t: \hat{r}_{k,t}(l) = \hat{x}_{k,t}^\prime(l)\Big(1-v_{k,t}^r(l) \sum_{n=1}^N v_{n,k}^h(l) v_{n,t}^s(l) \Big) + v_{k,t}^r(l) \sum_{n=1}^N \hat{h}_{n,k}^*(l) \hat{s}_{n,t}(l) $ & (A12) \\
		\ \ \quad $\forall n, k: \hat{h}_{n,k}(l+1) = \mathsf{E}[h_{n,k}|\hat{q}_{n,k}(l), v_{n,k}^q(l) ]$ & (A13) \\
		\ \ \quad $\forall n, k: v_{n,k}^h(l+1) = \mathsf{E}[|h_{n,k}-\hat{h}_{n,k}(l+1)|^2|\hat{q}_{n,k}(l), v_{n,k}^q(l)] $ & (A14) \\
		\ \ \quad $\forall k, t: \hat{x}_{k,t}^\prime(l+1) = \mathsf{E}[x_{k,t}^\prime|\hat{r}_{k,t}(l), v_{k,t}^r(l) ]$ & (A15) \\
		\ \ \quad $\forall k, t: v_{k,t}^x(l+1) = \mathsf{E}[|x_{k,t}^\prime-\hat{x}_{k,t}^\prime(l+1)|^2|\hat{r}_{k,t}(l), v_{k,t}^r(l) ] $ & (A16) \\
		\ \ \quad \textbf{if} $\sum_{n,t} |\bar{p}_{n,t}(l) - \bar{p}_{n,t}(l-1) |^2 < \epsilon \sum_{n,t}|\bar{p}_{n,t}(l) |^2 $, \textbf{stop}\\
		\ \ \textbf{end} \\
		\ \ $\hat{x}_{k,t}^\prime(1) = \hat{x}_{k,t}^\prime(l+1); v_{k,t}^x(1)= v_{k,t}^x(l+1); \hat{h}_{n,k}(1) = 0; v_{n,k}^h = 1$.\\
		\textbf{end} \\
	\hline
	\end{tabular}
	\end{table}

	The BiG-AMP algorithm is given in Table \ref{tab:algorithm}. The details of the algorithm are explained based on the factor graph in Fig. \ref{Fig:FactorGraph} as follows. In (A1)-(A2) of Table \ref{tab:algorithm}, messages from variable nodes $\{z_{n,k,t}\}$ to check nodes $\{\mathtt{c}_{n,t} \}$ are cumulated to obtain an estimate of $\mathbf{HX}^\prime$ with means $\{\bar{p}_{n,t}(l)\}$ and variances $\{\bar{v}_{n,t}^p(l)\}$. In (A3)-(A4), ``Onsager'' correction is applied to generate means $\{\hat{p}_{n,t}(l)\}$ and variances $\{v_{n,t}^p(l)\}$. More discussions on ``Onsager'' terms can be found in \cite{Donoho10}. In (A5)-(A6), the means $\{\hat{z}_{n,t}(l)\}$ and variances $\{v_{n,t}^z(l)\}$ are calculated based on $\{\hat{p}_{n,t}(l)\}$, $\{v_{n,t}^p(l)\}$, and the channel observations $\{y_{n,t}^\prime\}$. In (A7)-(A8), the scaled residual $\{\hat{s}_{n,t}(l)\}$ and a set of inverse-residual-variances $\{v_{n,t}^s(l)\}$ are computed. Then, each pair of $(\hat{s}_{n,t}(l), v_{n,t}^s(l))$ is sent from check node $\mathtt{c}_{n,t}$ to variable nodes $\{z_{n,k,t}|\forall k\}$ and then passed to the check nodes $\{\mathtt{c}_{n,k,t}|\forall k \}$.  In (A9)-(A10), messages from check nodes $\{\mathtt{c}_{n,k,t}\}$ to variable node $h_{n,k}$ are combined to compute an estimate of $h_{n,k}$, denoted by $\hat{q}_{n,k}(l)$, and the corresponding variance $v_{n,k}^q(l)$. Then, each pair of $(\hat{q}_{n,k}(l), v_{n,k}^q(l))$ is merged with the prior distribution $p_{h_{n,k}}(h_{n,k})$ to produce the posterior mean $\hat{h}_{n,k}(l+1)$ in (A13) and the variance $v_{n,k}^h(l+1)$ in (A14). A similar process is performed for $\{x_{k,t}^\prime\}$ in (A15)-(A16). Note that though not included in Table \ref{tab:algorithm}, damping is required to guarantee the convergence of BiG-AMP. We refer the interested readers to \cite{PSchniter14} for details.

\subsection{Ambiguity Elimination}
	With the BiG-AMP output $\hat{\mathbf{X}}^\prime$, we obtain an estimate of $\mathbf{X}$ as $\hat{\mathbf{X}} = \hat{\mathbf{X}}^\prime \mathbf{V}_1^\text{H}$. As aforementioned, we still need to eliminate the intrinsic ambiguities in $\hat{\mathbf{X}}$. The details are as follows.
	
	Recall from \eqref{equ:x_hat_model1} that $\hat{\mathbf{X}}$ can be modelled as $\hat{\mathbf{X}} = \mathbf{\Sigma} \mathbf{\Pi}  (\mathbf{X}  + \mathbf{\Delta})$. The permutation ambiguity $\mathbf{\Pi}$ can be resolved by inserting a label of $\lceil\log K \rceil$ bits for user identification. The phase ambiguity $\mathbf{\Sigma}$ can be resolved by using one element of each $\mathbf{x}_k$ to estimate $\Sigma_k$. To be specific, we assume without loss of generality that $\{x_{k,1}, k\in \mathcal{I}_K\}$ are known to the receiver. Recall that $\hat{x}_{k,1}$ can be expressed by
	\begin{equation}
		\hat{x}_{k,1} = \Sigma_k(x_{k,1}+\delta_{\pi(k),1} ).
	\end{equation}
	\begin{table}[!t]
	\caption{Algorithm 2: The P-BiG-AMP Algorithm}
	\label{tab:P-BiG-AMP}
	\centering
	\begin{tabular}{l r}
	\hline
		\textbf{Input}: Received signal $\mathbf{Y}$, prior distributions $p_{\mathbf{H}}(\mathbf{H}) $ and $p_{\mathbf{X}}(\mathbf{X})$.\\
		step 1: Perform SVD on $\mathbf{Y}$ to estimate the signal subspace: $\mathbf{Y} = \mathbf{UDV}^\text{H}$ with $\mathbf{V} = [\mathbf{V}_1, \mathbf{V}_2]$;\\
		step 2: Project $\mathbf{Y}$ onto $\mathbf{V}_1$: $\mathbf{Y}^\prime = \mathbf{Y}\mathbf{V}_1$;\\
		step 3: Factorize $\mathbf{Y}^\prime$ to obtain $(\hat{\mathbf{H}},\hat{\mathbf{X}}^\prime)$ based on the BiG-AMP algorithm;\\
		step 4: Calculate $\hat{\mathbf{X}} = \hat{\mathbf{X}}^\prime \mathbf{V}_1^\text{H}$;\\
		step 5: Calculate $\hat{\Sigma}_k = \frac{x_{k,1}}{|x_{k,1}|^2+\sigma_{\delta_{\pi(k)}}^2 } \hat{x}_{k,1}, k \in \mathcal{I}_K$;\\
		step 6: Output the estimate of the $k$th column of $\mathbf{X}$ as $\bar{\mathbf{x}}_k = \hat{\Sigma}_k^{-1} \hat{\mathbf{x}}_k, k\in \mathcal{I}_K$.\\
	\hline
	\end{tabular}
	\end{table}
	Then, an estimate of $\Sigma_k$ is given by
	\begin{equation}
		\hat{\Sigma}_k = \frac{x_{k,1}}{|x_{k,1}|^2+\sigma_{\delta_{\pi(k)}}^2 } \hat{x}_{k,1},
	\end{equation}
	where $\sigma_{\delta_k}^2$ is the average variance of $\pmb{\delta}_k$. Denote by $\tilde{\mathbf{x}}_k = [x_{k,2},\cdots, x_{k,T}]^\text{T}$ and $\hat{\tilde{\mathbf{x}}}_k = [\hat{x}_{k,2},\cdots, \hat{x}_{k,T}]^\text{T}$. Then, the estimate of each $\tilde{\mathbf{x}}_k$ is given by
	\begin{equation}
		\bar{\mathbf{x}}_k = \hat{\Sigma}_k^{-1} \hat{\tilde{\mathbf{x}}}_k.
	\end{equation}
	The overall P-BiG-AMP algorithm is summarized in Table \ref{tab:P-BiG-AMP}. The corresponding normalized mean-square-error (MSE) of $\mathbf{X}$ is given by
	\begin{equation}
		\mathsf{MSE}_{\mathbf{X}} = \frac{1}{K}\sum_{k=1}^K \frac{\mathsf{E}[\|\tilde{\mathbf{x}}_k -\bar{\mathbf{x}}_k \|_2^2]}{\mathsf{E}[\|\tilde{\mathbf{x}}_k\|_2^2]}.\label{equ:NMSEx}
	\end{equation}
	Likewise, an achievable rate of the massive MIMO system is given by 
	\begin{equation}
		\mathsf{R}_{\text{blind}} = \sum_{k=1}^K \left(1-\frac{1}{T}\right) \log (1+ \mathsf{SNR}_{k,\text{out}}) -  \frac{K\lceil\log K \rceil }{T},
	\end{equation}
	where the output SNR for each transmit terminal $k$ is given by 
	\begin{equation}
		\mathsf{SNR}_{k,\text{out}} = \frac{\mathsf{E}[\|\tilde{\mathbf{x}}_k\|_2^2]}{\mathsf{E}[\|\tilde{\mathbf{x}}_k -\bar{\mathbf{x}}_k \|_2^2]},
	\end{equation}
	and the rate loss $\frac{K\lceil\log K \rceil }{T}$ is caused by the permutation ambiguity.
		
\section{Numerical Results}
	In simulation, we set $\alpha_k = 1/K, \forall k\in \mathcal{I}_K$, and $P = K$, i.e., each element of $\mathbf{X}$ is independently and identically drawn from the standard Gaussian distribution. The SNR is given by $\frac{K}{\sigma^2}$. Following \cite{Selen07}, we assume that the channel matrix $\mathbf{H}$ in the angular domain follows the Bernoulli-Gaussian (BG) distribution
	\begin{equation}
		p_{h_{n,k}}(h) = (1-\rho)\delta(h) + \rho \mathcal{CN}(h; 0, 1), \text{ for all } n \in \mathcal{I}_N, k\in \mathcal{I}_K.
	\end{equation} 
	For the BiG-AMP algorithm, the maximum number of inner iteration $L_\text{max} = 100$, and the maximum number of outer iteration $M_\text{max} = 10$.
	
	We first consider the normalized MSE of both $\mathbf{X}$ and $\mathbf{H}$ of the proposed blind detection in Fig. \ref{Fig:MSE} with $N = 64, 128$, and $256$, $K = 8$, $T = 50$, and $\rho = 0.2$. The normalized MSE of $\mathbf{X}$ is defined in \eqref{equ:NMSEx} and the normalized MSE of $\mathbf{H}$ is defined as $\frac{\|\mathbf{H}-\hat{\mathbf{H}}\|_F^2}{\|\mathbf{H}\|_F^2}$ for any estimate $\hat{\mathbf{H}}$. Clearly, the normalized MSE of both $\mathbf{X}$ and $\mathbf{H}$ decreases as $N$ increases. From Fig. \ref{Fig:MSE}, we see that our proposed blind detection scheme works well under practical antenna setups. 
	
	\begin{figure}[!t]
		\centering
		\includegraphics[width = 0.7\textwidth]{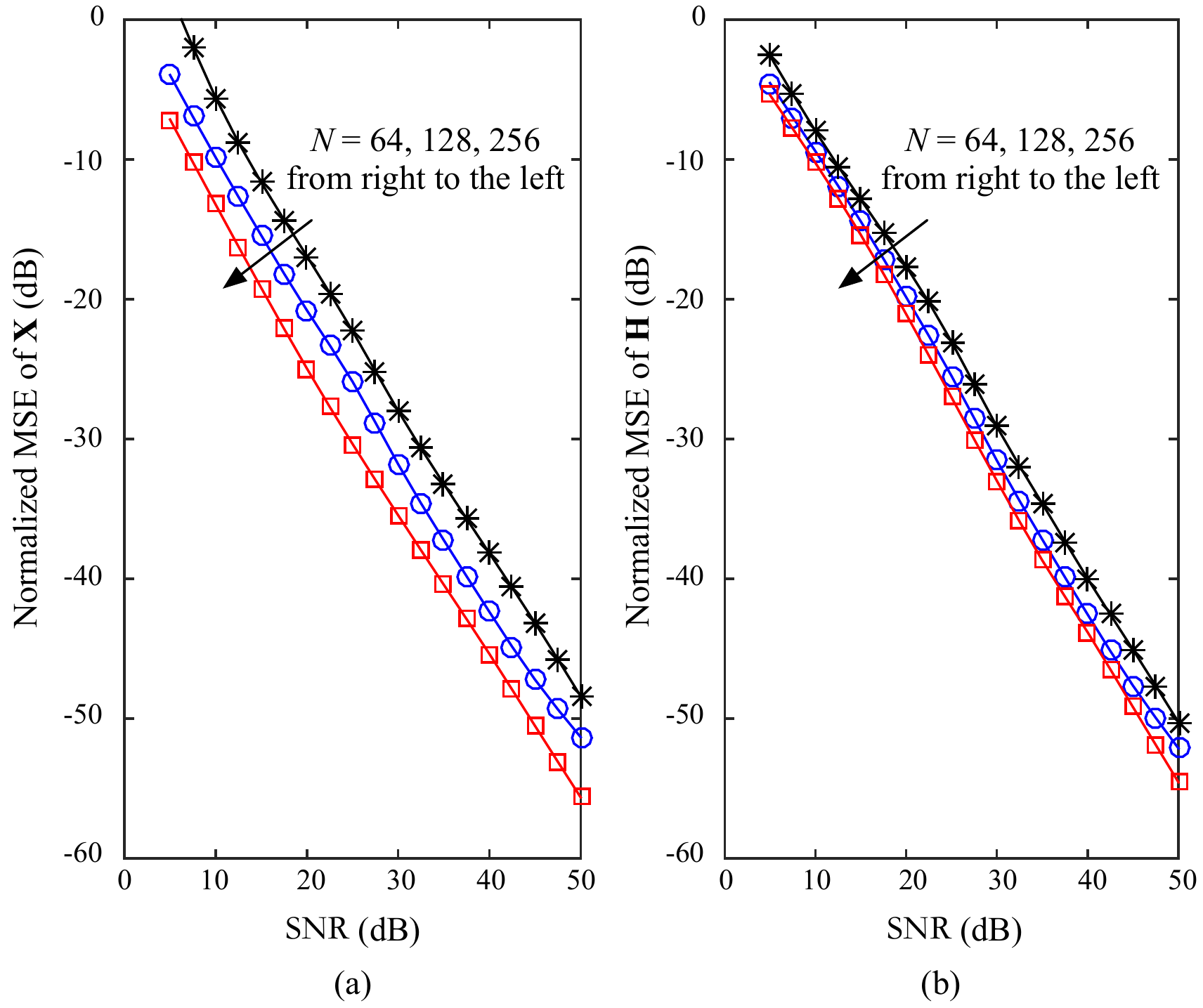}
	 	\caption{Normalized MSE of the proposed blind detection versus SNR with $N = 64, 128$, and $256$, $K = 8$, $T = 50$, and $\rho = 0.2$. The normalized MSE of $\mathbf{X}$ is shown in (a) and the normalized MSE of $\mathbf{H}$ is given in (b). }\label{Fig:MSE}
	\end{figure}

	\begin{figure}[!t]
		\centering
		\includegraphics[width = 0.7\textwidth]{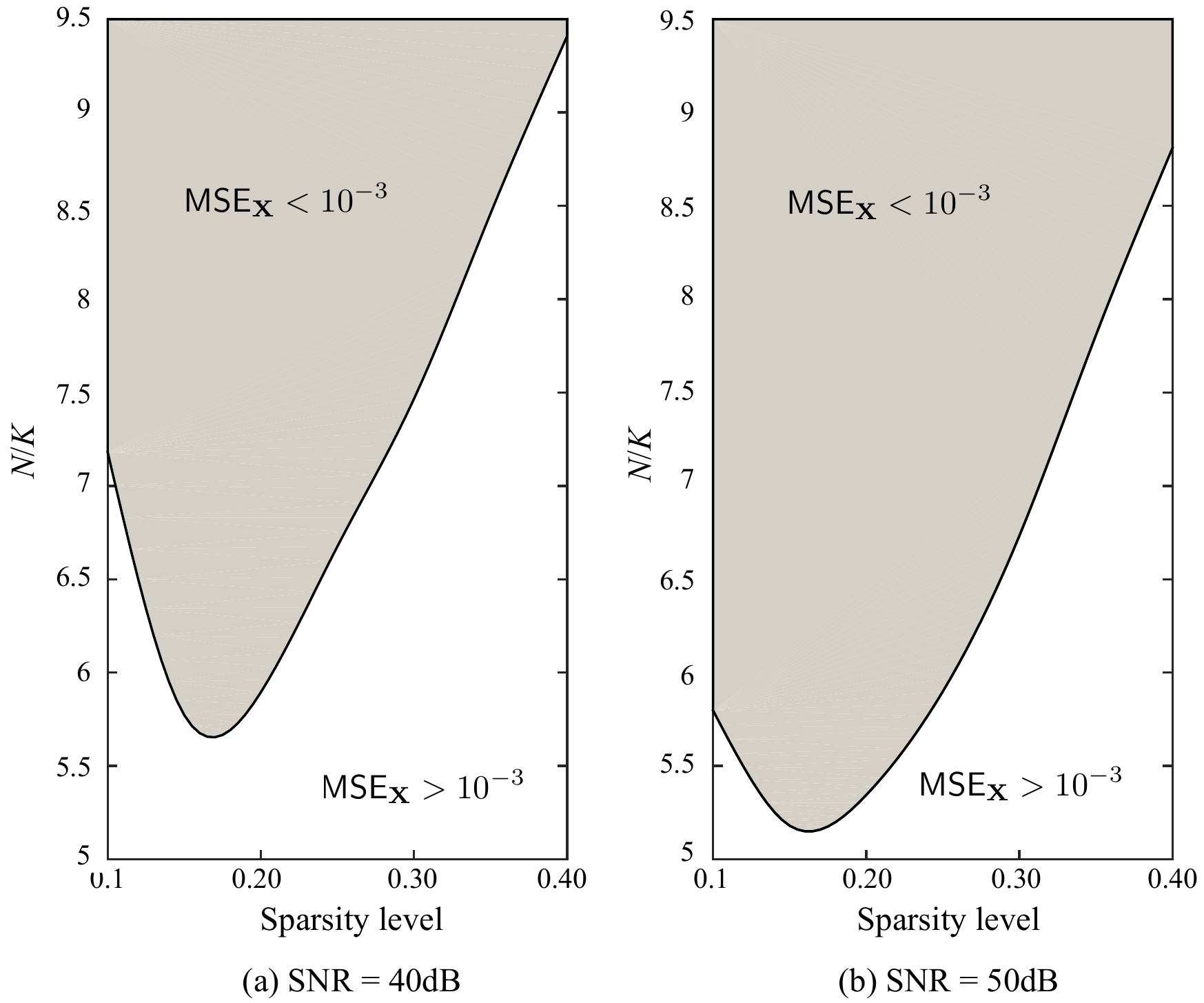}
	 	\caption{The phase transition of blind detection with sparsity level $\rho$ and the ratio $N/K$. The other system parameters are set as $K  = 20, T = 100$, SNR = 40dB in (a) and SNR = 50dB in (b).}\label{Fig:FeasibleRegion}
	\end{figure}
	
	As stated in Theorem \ref{theorem:DoF}, the blind signal detection scheme can achieve a DoF close to the ideal one with perfect channel knowledge, provided that $N$ is sufficiently large and $\rho$ is sufficiently small. However, the requirements on $N$ and $\rho$ to ensure successful detection in Theorem \ref{theorem:DoF} are given in the form of asymptotic functions of $K$ and $T$, and are difficult to evaluate for finite values of $K$ and $T$. Here, we examine through numerical simulations the tradeoff between $N$ and $\rho$ to guarantee detection success when $K$ and $T$ are finite. In particular, we say that the detection is successful and the corresponding values of $N$ and $\rho$ are feasible if $\mathsf{MSE}_{\mathbf{X}} < 10^{-3}$. Fig. \ref{Fig:FeasibleRegion} shows the feasible region of $(\rho, N/K)$ for $K = 20$ and $T = 100$. The SNR is set to 40dB in the left subgraph and 50dB in the right, and the simulation is taken over $10^4$ channel realizations. We see that the feasible region becomes larger as the SNR increases. This is reasonable since a higher SNR implies a better channel quality, and hence a less stringent requirement on $N$ and $\rho$ to guarantee successful detection. We also see that the boundary curve is not a monotonic function of the sparsity level $\rho$. Intuitively, the smaller the sparsity level $\rho$, the more sparse the channel, the less the channel randomness, and so the less stringent the requirement on $N$ to achieve successful detection. However, when $\rho$ is very close to zero, the probability of sparsity pattern collision increases, leading to an increase in the required $N$ to guarantee successful detection.
	
	\begin{figure}[!t]
		\centering
		\includegraphics[width = 0.7\textwidth]{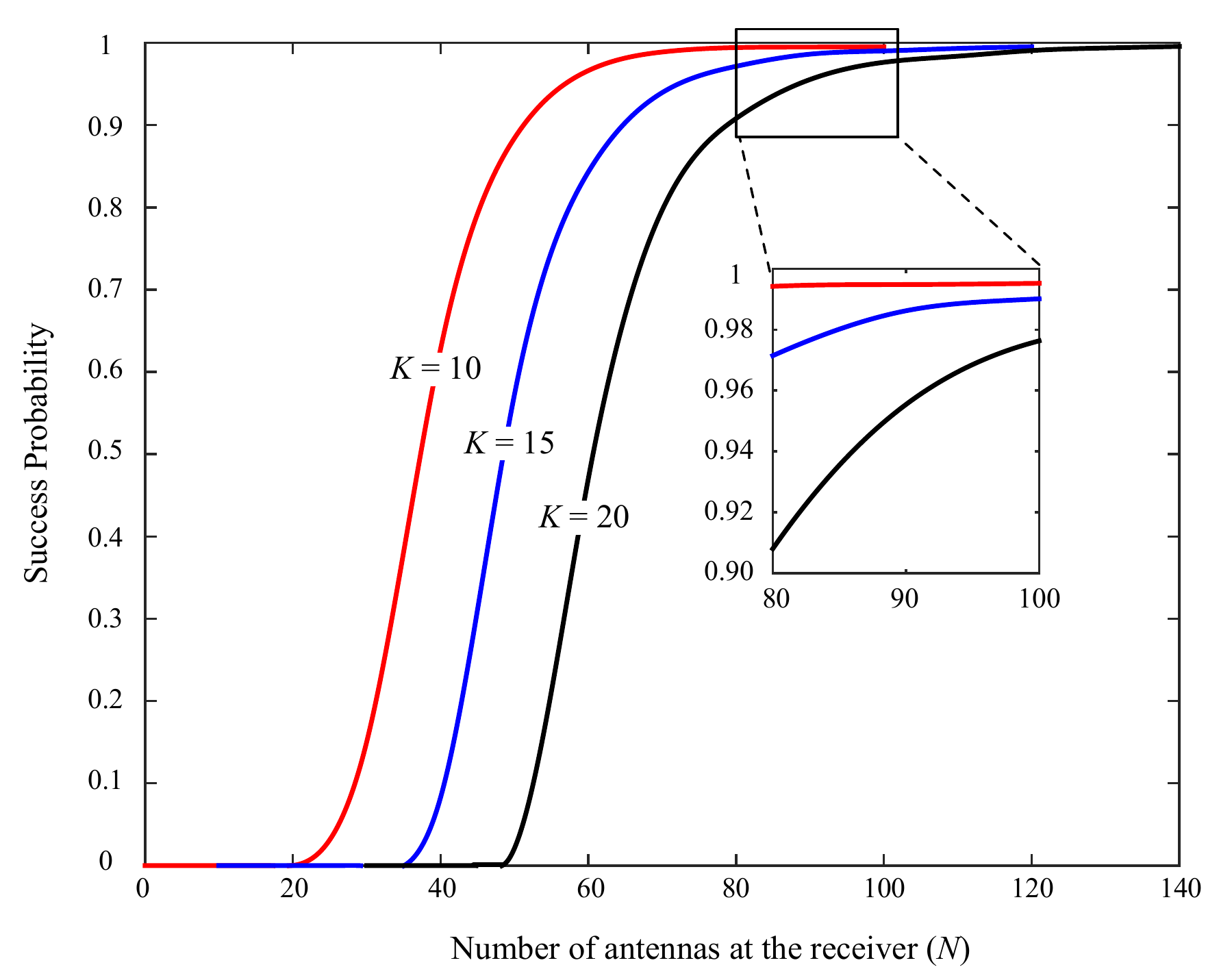}
	 	\caption{The success probability of blind detection versus the number of antennas $N$ at the receiver with the number of transmit terminals $K = 10, 15, \text{ and } 20$, respectively. The other system settings are $T = 100, \rho = 0.2$, and SNR = 40dB.}\label{Fig:SuccessProb}
	\end{figure}
	
	Fig. \ref{Fig:SuccessProb} plots the success probability of our proposed blind detection against the number of receive terminals $N$ with $K = 10, 15,$ and $20$, respectively. The other system settings are $T = 100, \rho = 0.2$, and SNR = 40dB. The success probability is calculated based on $10^5$ channel realizations. From Fig. \ref{Fig:SuccessProb}, we see that the success probability increases with $N$, and approaches one for a sufficiently large $N$. This is because the probability of sparsity pattern collision is non-zero for finite values of $N, K, T, \rho$, and vanishes as $N$ goes to infinity, as stated in Theorem \ref{theorem:DoF}. 

	\begin{figure}[!t]
		\centering
		\includegraphics[width = 0.7\textwidth]{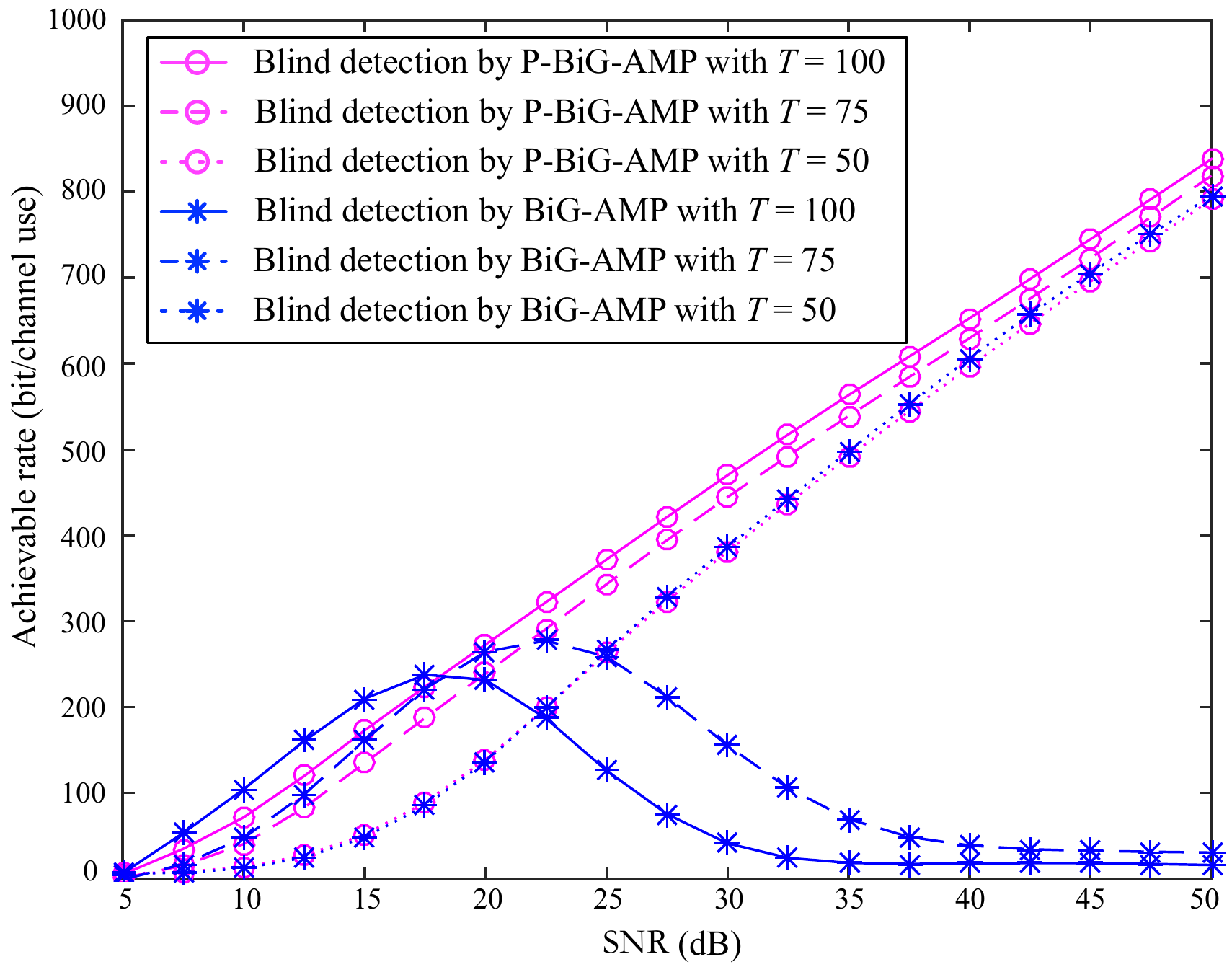}
	 	\caption{The performance of the proposed blind detection scheme with various $T$. The system settings are given by $N = 500, K  = 50, \text{ and } \rho = 0.3$.}\label{Fig:PerfCompvsT}
	\end{figure}
	
	Fig. \ref{Fig:PerfCompvsT} shows the achievable rate of the blind detection scheme with the proposed P-BiG-AMP scheme for different values of $T$. The coherence time slot $T$ is set to 50, 75, and 100. Moreover, $N = 500, K = 50,$ and $\rho = 0.3$. The simulation is taken over 200 channel realizations. From Fig. \ref{Fig:PerfCompvsT}, we see that the achievable rate of P-BiG-AMP monotonically increases in both $T$ and SNR. For comparison, the achievable rate of the blind detection scheme with BiG-AMP in \cite{PSchniter14} is also included. We see that the achievable rate of BiG-AMP monotonically increases in both $T$ and SNR in the relatively low SNR regime, while for $T > K$ (i.e., $T = 75$ and $100$), the performance of BiG-AMP deteriorates as SNR further increases. This demonstrates the advantage of the proposed P-BiG-AMP algorithm.
	
	\begin{figure}[!t]
		\centering
		\includegraphics[width = 0.7\textwidth]{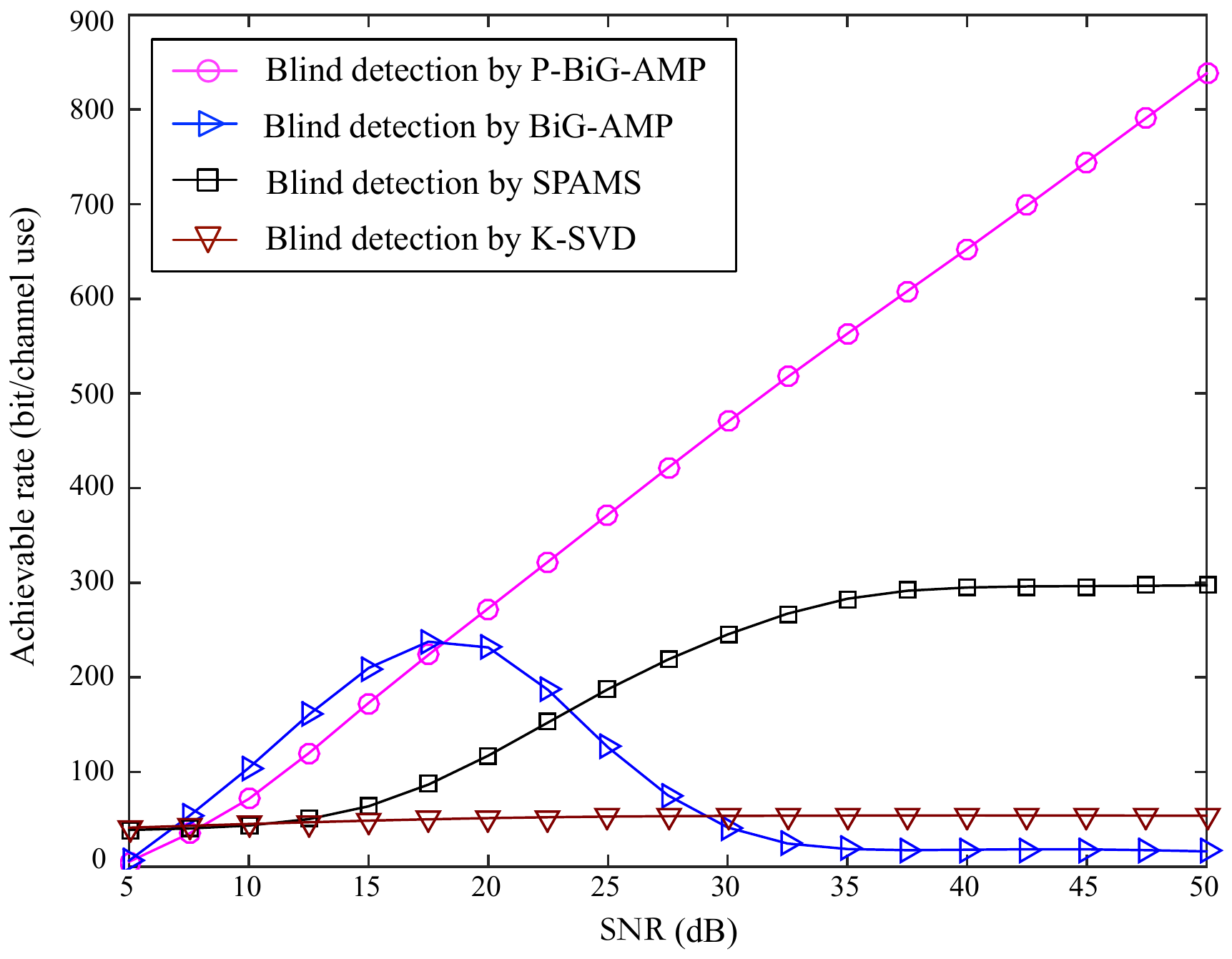}
	 	\caption{The performance of the blind detection scheme with various detection algorithms. The system settings are given by $N = 500, K  = 50, \text{ and } \rho = 0.3$.}\label{Fig:DetAlgComp}
	\end{figure}
	
	In Fig. \ref{Fig:DetAlgComp}, we compare the achievable rate of the blind detection with various dictionary learning algorithms, including the K-SVD algorithm \cite{KSVD06}, the SPAMS algorithm \cite{SPAMS10}, the BiG-AMP algorithm\cite{PSchniter14}, and the P-BiG-AMP algorithm. The system settings are $N = 500, K = 50, T = 100,$ and $\rho = 0.3$. The maximum iteration number for K-SVD is 100, and that for SPAMS is 1000. We see that the P-BiG-AMP algorithm significantly outperforms its counterparts especially in the high SNR regime.
	
	We are now ready to compare the achievable rate of the proposed P-BiG-AMP based blind detection scheme with other approaches for massive MIMO, as listed below.
	\begin{enumerate}
  		\item [i).] Training-based MIMO coherence detection \cite{Yuan16,Hassibi03}: Each transmission frame consists of two phases, namely, the training phase for channel estimation based on the pilot signals and the data transmission phase for data detection based on the estimated channel and the received signals. Channel sparsity is not taken into consideration.
  		\item [ii).] Blind detection without exploiting channel sparsity \cite{Zheng02}: No pilots are required and the data is detected by sphere packing in the Grassmann manifold. 
  		\item [iii).] Compressed-sensing based MIMO coherence detection \cite{Bajwa10}: Channel sparsity is exploited to reduce the number of required pilots. Compressed sensing algorithms are used to estimate channel coefficients based on pilots. In our simulation, approximate message passing (AMP) algorithm \cite{Donoho10} is used to solve the compressed sensing problem.
  		\item [iv).] Blind channel estimation based on channel sparsity \cite{Mezghani16}: $\mathbf{H}$ is estimated based on the covariance of $\mathbf{Y}$ and the channel sparsity, and then $\mathbf{X}$ is detected based on the estimated channel. The channel estimation and the data detection are separated from each other.
  		\item [v).] Blind detection by exploiting channel sparsity: We jointly estimate the channel and the data from the received signal based on the P-BiG-AMP algorithm.
	\end{enumerate}
	\begin{figure}[!t]
		\centering
		\includegraphics[width = 0.7\textwidth]{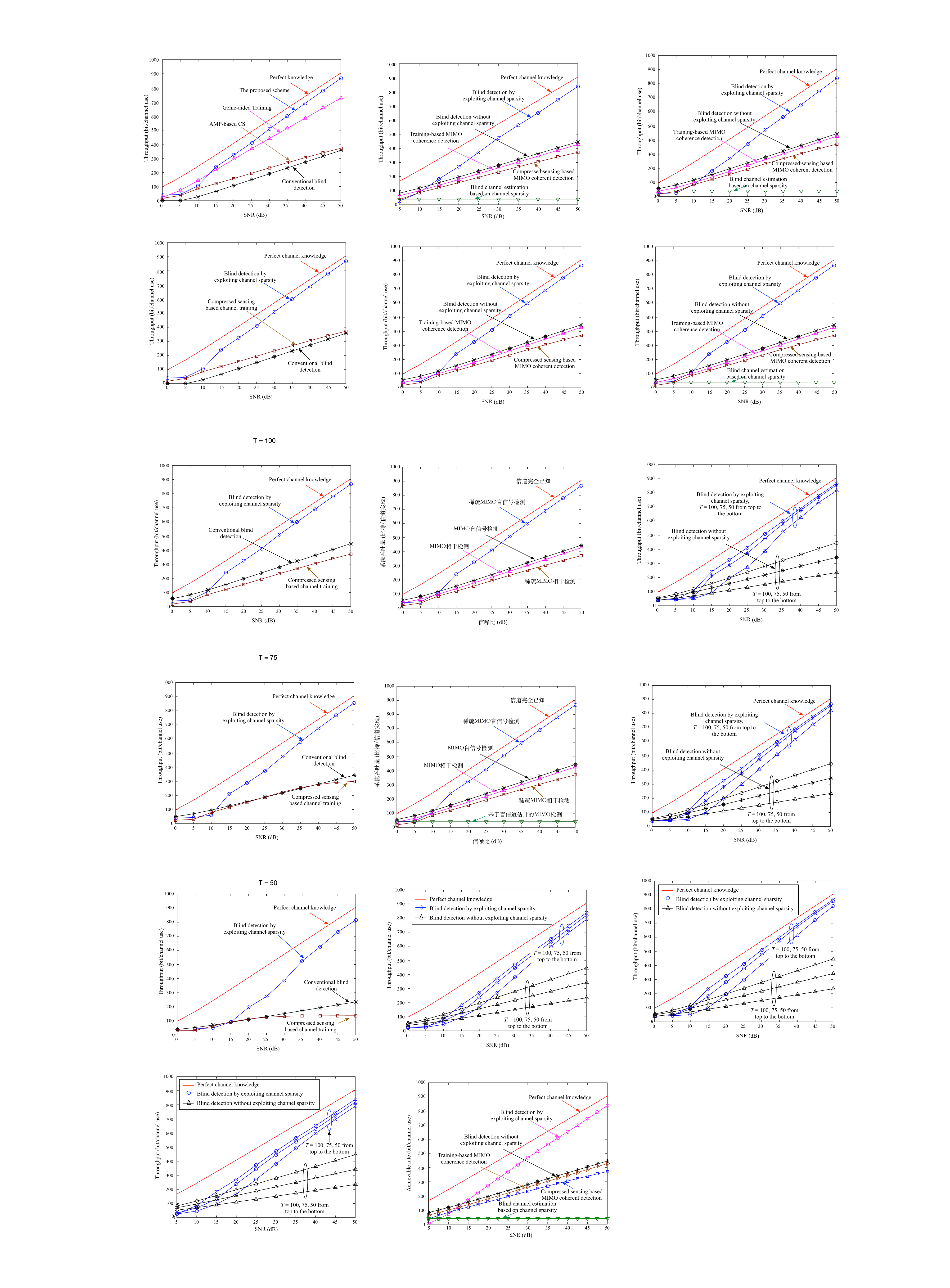}
	 	\caption{Performance comparison among various schemes with $N = 500, K = 50, T = 100$, and $ \rho = 0.3$. }\label{Fig:PerfComp}
	\end{figure}
	
	\begin{figure}[!h]
		\centering
		\includegraphics[width = 0.7\textwidth]{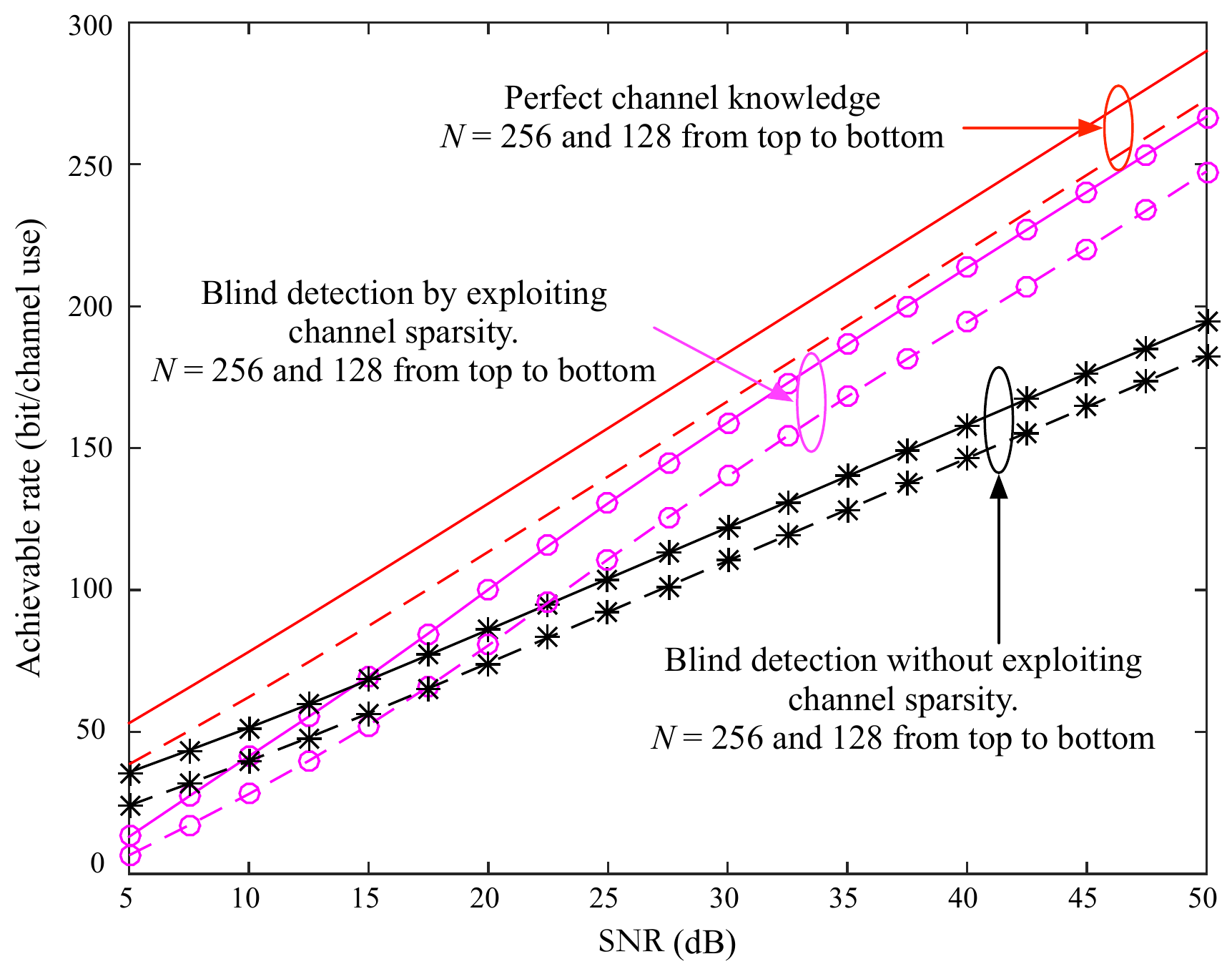}
	 	\caption{Performance comparison among various schemes with $N = 128$ and $256$, $K = 16$, $T = 50$, and $\rho = 0.2$.}\label{Fig:Perf2}
	\end{figure}
	
	Fig. \ref{Fig:PerfComp} gives the performance comparison of the above mentioned schemes. The ideal case with perfect channel knowledge at the receiver is also included for comparison. In Fig. \ref{Fig:PerfComp}, we set $N = 500, K = 50, T = 100$, and $\rho = 0.3$. The simulation results are taken over $200$ channel samples. From Fig. \ref{Fig:PerfComp}, we see that the proposed blind detection scheme significantly outperforms the other existing schemes, and performs close to the ideal case with perfect channel knowledge in the high SNR regime. This is in agreement with the DoF result in Section IV. We note that the achievable rate of the blind channel estimation scheme does not increase with SNR. The reason is that $T=100$ is too small and the covariance matrix of $\mathbf{Y}$ is poorly estimated.
	
	Fig. \ref{Fig:Perf2} shows the simulation results for a more practical antenna configuration of $N = 128$ and $256$. For comparison, the performance of the ideal case with perfect channel knowledge and the blind detection scheme without exploiting channel sparsity in \cite{Zheng02} are also included. From Fig. \ref{Fig:Perf2}, we see that the proposed blind detection scheme still performs close to the ideal case and outperforms the blind detection scheme without exploiting channel sparsity with practical $N$.

\section{Conclusions}
	In this paper, we investigated the impact of the channel sparsity on the fundamental performance limit of the massive MIMO system. We developed a novel blind detection scheme to efficiently exploit the channel sparsity inherent in massive MIMO systems. The proposed blind detection scheme simultaneously estimates the channel and the signal by directly factorizing the received signal matrix. We showed that the proposed scheme achieves a DoF arbitrarily close to $K(1-1/T)$ provided that $N$ is sufficiently large and $\rho$ is sufficiently small. Such an achievable DoF is very close to the ideal DoF with perfect CSI at the receiver. Moreover, we proposed the P-BiG-AMP algorithm to study the performance of the blind detection scheme in the finite SNR regime. Numerical results show that, in the medium to high SNR regime, the proposed scheme achieves a much higher throughput than the counterpart schemes under various system configurations.
	 

\appendices

\section*{Acknowledgment}
	The authors would like to thank Prof. Yang Yang for insightful discussions and the anonymous reviewers for their constructive comments to significantly improve the presentation of this paper.

\ifCLASSOPTIONcaptionsoff
  \newpage
\fi


%


\begin{thebibliography}{10}
\bibitem{Larsson14} 
E. G. Larsson, O. Edfors, F. Tufvesson, and T. L. Marzetta, ``Massive MIMO for next generation wireless systems,'' \emph{IEEE Commun. Mag.}, vol. 52, no. 2, pp. 186-195, Feb. 2014.

\bibitem{Marzetta10} 
T. L. Marzetta, ``Noncooperative cellular wireless with unlimited numbers of base station antennas,'' \emph{IEEE Trans. Wireless Commun.}, vol. 9, no. 11, pp. 3590-3600, Nov. 2010.

\bibitem{Ngo13} 
H. Q. Ngo, E. G. Larsson, and T. L. Marzetta, ``Energy and spectral efficiency of very large multiuser MIMO systems,'' \emph{IEEE Trans. Commun.}, vol. 61, no. 4, pp. 1436-1449, Apr. 2013.

\bibitem{Hoydis13} 
J. Hoydis, S. T. Brink, and M. Debbah, ``Massive MIMO in the UL/DL of cellular networks: How many antennas do we need?'' \emph{IEEE J. Sel. Areas Commun.}, vol. 31, no. 2, pp. 160-171, Feb. 2013.

\bibitem{Geraci13}
G. Geraci, R. Couillet, J. Yuan, M. Debbah, and I. B. Collings, ``Large system analysis of linear precoding in MISO broadcast channels with confidential messages,'' \emph{IEEE J. Sel. Areas Commun.}, vol. 31, no. 9, pp. 1660-1671, Sept. 2013.

\bibitem{Fang16}
Z. Fang, X. Yuan, X. Wang, and C. Li, ``Nonregenerative cellular two-way relaying with large-scale antenna arrays,'' \emph{IEEE Trans. on Veh. Technology}, vol. 65, no. 7, pp. 4959-4972, Jul. 2016.

\bibitem{ITU}
M.2083 : IMT Vision, ``Framework and overall objectives of the future development of IMT for 2020 and beyond,'' Sep. 2015.

\bibitem{Wang}
C. Wang, F. Haider, X. Gao, X. You, Y. Yang, D. Yuan, H. Aggoune, H. Haas, S. Fletcher, and E. Hepsaydir, ``Cellular architecture and key technologies for 5G wireless communication networks,'' \emph{IEEE Communications Magazine}, vol. 52, no. 2, pp. 122-130, Feb. 2014.

\bibitem{Coldrey07} 
M. Coldrey (Tapio) and P. Bohlin, ``Training-based MIMO systems-part I: performance comparison,'' \emph{IEEE Trans. Signal Process.}, vol. 55, no. 11, pp. 5464-5476, Nov. 2007.

\bibitem{Yuan16} 
X. Yuan, C. Fan, and Y. Zhang, ``Fundamental limits of training-based multiuser MIMO systems,'' available at: http://arxiv.org/abs/1511.08977.

\bibitem{Hassibi03}
B. Hassibi and B. Hochwald, ``How much training is needed in multiple antenna wireless links?'' \emph{IEEE Trans. Inf. Theory}, vol. 49, no. 4, pp. 951-963, Apr. 2003.

\bibitem{Muquet02}
B. Muquet, M. De Courville, and P. Duhamel, ``Subspace-based blind and semi-blind channel estimation for OFDM systems,'' \emph{IEEE Trans. Signal Process.}, vol. 50, no. 7, pp. 1699-1712, Jul. 2002.

\bibitem{Ngo12}
H. Q. Ngo and E. G. Larsson, ``EVD-based channel estimation in multicell multiuser MIMO systems with very large antenna arrays,'' in \emph{Proc. IEEE Int. Conf. Acoust. Speech Signal Process. (ICASSP)}, 2012, pp. 3249-3252.

\bibitem{Zheng02}
L.~Zheng, and D.~N.~C.~Tse, ``Communication on the Grassmann manifold: A geometric approach to the noncoherent multiple-antenna channel,'' \emph{IEEE Trans. Inf. Theory}, vol. 48, no. 2, pp. 359-383, Feb. 2002.

\bibitem{Sayeed07}
Z.~Yan, M.~Herdin, A.~M.~Sayeed, and E.~Bonek, ``Experimental study of MIMO channel statistics and capacity via the virtual channel representation,'' Univ. Wisconsin-Madison, Madison, Tech. Rep., Feb. 2007.

\bibitem{Sayeed02}
A.~M.~Sayeed, ``Deconstructing multiantenna fading channels,'' \emph{IEEE Trans. Signal Process.}, vol. 50, no. 10, pp. 2563-2579. Oct. 2002.

\bibitem{Samimi16}
M. K. Samimi and T. S. Rappaport, ``3-d millimeter-wave statistical channel model for 5G wireless system design,'' \emph{IEEE Trans. Microwave Theory and Techniques}, vol. 64, no. 7, pp. 2007-2225, Jul. 2016.

\bibitem{Vuokko07}
L. Vuokko, V.-M. Kolmonen, J. Salo, and P. Vainikainen, ``Measurement of large-scale cluster power characteristics for geometric channel models,'' \emph{IEEE Trans. Antennas and Propagation}, vol. 55, no. 11, pp. 3361-3365, Nov. 2007.

\bibitem{Czink07}
N. Czink, X. Yin, H. Ozcelik, M. Herdin, E. Bonek, and B. H. Fleury, ``Cluster characteristics in a MIMO indoor propagation environment,'' \emph{IEEE Trans. Wireless Commun.}, vol. 6, no. 4, pp. 1465-1475, Apr. 2007.

\bibitem{Yin13}
H. Yin, D. Gesbert, M. Filippou, and Y. Liu, ``A coordinated approach to channel estimation in large-scale multiple antenna systems,'' \emph{IEEE J. Sel. Areas Commun.}, vol. 31, no. 2, pp. 264-273, Feb. 2013.

\bibitem{Bajwa10}
W.~U.~Bajwa, J.~Haupt, A.~M.~Sayeed, and R.~Nowak, ``Compressed channel sensing: A new approach to estimating sparse multipath channels,'' \emph{IEEE Proceedings}, vol. 98, no. 6, pp. 1058-1076, Jun. 2010.

\bibitem{Lau14}
X.~Rao, and V.~K.~N.~Lau, ``Distributed compressive CSIT estimation and feedback for FDD multi-user massive MIMO systems,'' \emph{IEEE Trans. Signal Process.}, vol. 62, no. 12, pp. 3261-3271, Jun. 2014.

\bibitem{Masood15}
M. Masood, L. H. Afify, and T. Y. Al-Naffouri, ``Efficient coordinated recovery of sparse channels in massive MIMO,'' \emph{IEEE Trans. Signal Process.}, vol. 63, no. 1, pp. 104-118, Jan. 2015.

\bibitem{Muller14}
R. R. Muller, L. Cottatellucci, and M. Vehkapera, ``Blind pilot decontamination,'' \emph{IEEE J. Sel. Topics Signal Process.}, vol. 8, no. 5, pp. 773-786, Oct. 2014.

\bibitem{Mezghani16}
A. Mezghani and A. L. Swindlehurst, ``Blind estimation of sparse multi-user massive MIMO channels,'' arXiv:1612.00131v1, Dec. 2016.

\bibitem{Sun17}
J. Sun, Q. Qu, and J. Wright, ``Complete dictionary recovery over the sphere I: Overview and the geometric picture,'' \emph{IEEE Trans. Inform. Theory}, vol. 63, no. 2, pp. 853-884, Feb. 2017.

\bibitem{Koren09}
Y. Koren, R. Bell, and C. Volinsky, ``Matrix factorization techniques for recommender systems,'' \emph{Computer}, vol. 42, no. 8, Aug. 2009.

\bibitem{KSVD06}
M. Aharon, M. Elad, and A. Bruckstein, ``K-SVD: An algorithm for designing overcomplete dictionaries for sparse representation,'' \emph{IEEE Trans. Signal Process.}, vol. 54, no. 11, pp. 4311-4322, 2006.

\bibitem{SPAMS10}
J. Mairal, F. Bach, J. Ponce, and G. Sapiro, ``Online learning for matrix factorization and sparse coding,'' \emph{J. Mach. Learn. Res.}, vol. 11, pp. 19-60, 2010.

\bibitem{Spielman12}
D. Spielman and H. Wang and J. Wright, ``Exact recovery of sparsely-used dictionaries,'' \emph{preprint}, available at: http://www.columbia.edu/~jw2966, 2012.

\bibitem{PSchniter14}
J. T. Parker, P. Schniter, and V. Cevher, ``Bilinear generalized approximate message passing – Part I: Derivation,'' \emph{IEEE Trans. Signal Process.}, vol. 62, no. 22, pp. 5839-5853, Nov. 2014.

\bibitem{Telatar99}
I. Telatar, ``Capacity of multi-antenna Gaussian channels,'' \emph{European Trans. Telecommun.}, vol.10, pp. 585-595, Nov. 1999.

\bibitem{Donoho10}
D. L. Donoho, A. Maleki, and A. Montanari, ``Message-passing algorithms for compressed sensing,'' \emph{Proceedings of the National Academy of Sciences of the United States of America}, vol. 106, no. 45, pp. 18914-18919, 2010.

\bibitem{Candes09}
E. J. Candes and Y. Plan, ``Near-ideal model selection by $\ell_1$ minimization,'' \emph{Ann. Statist.}, vol. 37, no. 5A, pp. 2145-2177, 2009.

\bibitem{SLearning}
T. Hastie, R. Tibshirani, and J. Friedman, \emph{The elements of statistical learning}, 2nd Edition, Springer Series in Statistics, 2008.

\bibitem{Gribonval15}
R. Gribonval, R. Jenatton, and F. Bach, ``Sparse and spurious: Dictionary learning with noise and outliers,'' \emph{IEEE Trans. Inf. Theory}, vol. 61, no. 11, pp. 6298-6319, Nov. 2015.

\bibitem{Medard00}
M. Medard, ``The effect upon channel capacity in wireless communications of perfect and imperfect knowledge of the channel,'' \emph{IEEE Trans. Inf. Theory}, vol. 46, no. 3, pp. 933-946, May 2000.

\bibitem{Selen07}
Y. Selen and E. G. Larsson, ``RAKE receiver for channels with a sparse impulse response,'' \emph{IEEE Trans. Wireless Commun.}, vol. 6, no. 9, pp. 3175-3180, Sep. 2007.




%
%





%
%
%
%
%
%
%
%
%
%
%
%
%
 
\end{thebibliography}
\end{document}